\documentclass[a4paper,UKenglish,cleveref,thm-restate]{lipics-v2021}
\pdfoutput=1 %
\hideLIPIcs  %
\nolinenumbers

\newcommand{\lipics}{}
\usepackage{ifthen}
\usepackage{comment}
\makeatletter
\ifthenelse{\not{\isundefined{\lipics}}}
{
\makeatletter
\let\c@author\relax
\makeatother
}{}
\usepackage[
style=numeric,
giveninits=true, 
maxbibnames=99, 
maxcitenames=2, 
natbib=true, 
labelnumber, 
defernumbers=true,  
backend=biber, 
sortcites=true,
doi=false,isbn=false,url=false
]{biblatex}
\renewbibmacro{in:}{%
  \ifentrytype{article}{}{\printtext{\bibstring{in}\intitlepunct}}
}
\let\oldcitet\citet
\renewcommand{\citet}[1]{\mbox{\oldcitet{#1}}}

\ifthenelse{\isundefined{\lipics}}
{
\usepackage[protrusion=true,expansion=true]{microtype}

\usepackage[colorlinks=true,bookmarks=false]{hyperref}
\usepackage[nameinlink,capitalise]{cleveref}
\usepackage[small,bf]{caption}
\usepackage{subcaption}
\usepackage[british]{babel}
\usepackage{paralist}
}
{}

\graphicspath{{img/}}
\usepackage{graphicx}

\usepackage[shortlabels]{enumitem}
\setlist[enumerate]{nosep} %
\setlist[itemize]{nosep} %

\usepackage{xcolor}
\definecolor{darkblue}{rgb}{0,0,0.45}
\definecolor{darkred}{rgb}{0.6,0,0}
\definecolor{darkgreen}{rgb}{0.13,0.5,0}
\hypersetup{colorlinks, linkcolor=darkblue, citecolor=darkgreen,
urlcolor=darkblue}

\ifthenelse{\isundefined{\llncs}}{
  \usepackage{amsthm}
}

\usepackage{amsmath,amsfonts,amssymb}
\usepackage{mathtools}
\usepackage{mathrsfs}
\usepackage{thm-restate}
\usepackage{wrapfig}
\usepackage{xspace}
\usepackage{csquotes}
\usepackage{multirow}

\usepackage{footmisc}

\crefname{observation}{Observation}{Observations}

\newcommand{\mc}[1]{{\mathcal{#1}}}

\newcommand{\eps}{{\varepsilon}}

\newcommand{\hy}{\hbox{-}\nobreak\hskip0pt}

\usepackage{tcolorbox}

\usepackage{environ,xstring}
\newif\iflabel
\newif\ifdbs
\newif\ifamp
\NewEnviron{doitall}{%
  \noexpandarg
  \expandafter\IfSubStr\expandafter{\BODY}{\label}{\labeltrue}{\labelfalse}%
  \expandafter\IfSubStr\expandafter{\BODY}{\\}{\dbstrue}{\dbsfalse}%
  \expandafter\IfSubStr\expandafter{\BODY}{&}{\amptrue}{\ampfalse}%
  \iflabel\def\doitallstar{}\else\def\doitallstar{*}\fi
  \ifdbs
    \ifamp
      \def\doitallname{align}%
    \else
      \def\doitallname{multline}%
    \fi
  \else
    \def\doitallname{equation}
  \fi
  \begingroup\edef\x{\endgroup
    \noexpand\begin{\doitallname\doitallstar}%
    \noexpand\BODY
    \noexpand\end{\doitallname\doitallstar}%
  }\x
}
\def\[#1\]{\begin{doitall}#1\end{doitall}}

\newcommand{\pname}[1]{\textsc{#1}}

\newcommand{\shs}{\pname{Sparse-HS}\xspace}
\newcommand{\svc}{\pname{Sparse-VC}\xspace}
\newcommand{\fvc}{\pname{Fair-VC}\xspace}
\newcommand{\ofvc}{\pname{Open-Fair-VC}\xspace}
\newcommand{\spc}{\pname{$r$-SPC}\xspace}
\newcommand{\hd}{\pname{$r$-HD}\xspace}
\newcommand{\remove}[1]{}

\usepackage{framed}
\newcommand{\executeiffilenewer}[3]{%
\ifnum\pdfstrcmp{\pdffilemoddate{#1}}%
{\pdffilemoddate{#2}}>0%
{\immediate\write18{#3}}\fi%
}

\bibliography{papers}

\newcommand{\polyn}{\cdot n^{O(1)}}

\DeclareMathOperator{\dist}{dist}

\title{On Sparse Hitting Sets:\newline 
from Fair Vertex Cover to Highway Dimension}

\titlerunning{On Sparse Hitting Sets:
from Fair Vertex Cover to Highway Dimension}

\author{Johannes Blum}{University of Konstanz, Germany  \and \url{https://algo.uni-konstanz.de/team/blum}}
{johannes.blum@uni-konstanz.de}{https://orcid.org/0000-0003-1102-3649}{}
\author{Yann Disser}{Technical University of Darmstadt, Germany}
{disser@mathematik.tu-darmstadt.de}{https://orcid.org/0000-0002-2085-0454}{}
\author{Andreas Emil Feldmann}{Charles University, Prague, Czechia \and 
\url{https://sites.google.com/site/aefeldmann/home}}{feldmann.a.e@gmail.com}
{https://orcid.org/0000-0001-6229-5332}
{Supported by the Czech Science Foundation GA{\v C}R (grant \#19-27871X).}
\author{Siddharth Gupta}{University of Warwick, United Kingdom \and \url{https://guptasid.bitbucket.io/}}
{siddharth.gupta.1@warwick.ac.uk}{https://orcid.org/0000-0003-4671-9822}{Supported by the Engineering and Physical Sciences Research Council (EPSRC) grant no: EP/V007793/1.}
\author{Anna Zych-Pawlewicz}{University of Warsaw, Poland}{anka@mimuw.edu.pl}
{https://orcid.org/0000-0002-5361-8969}{\flag{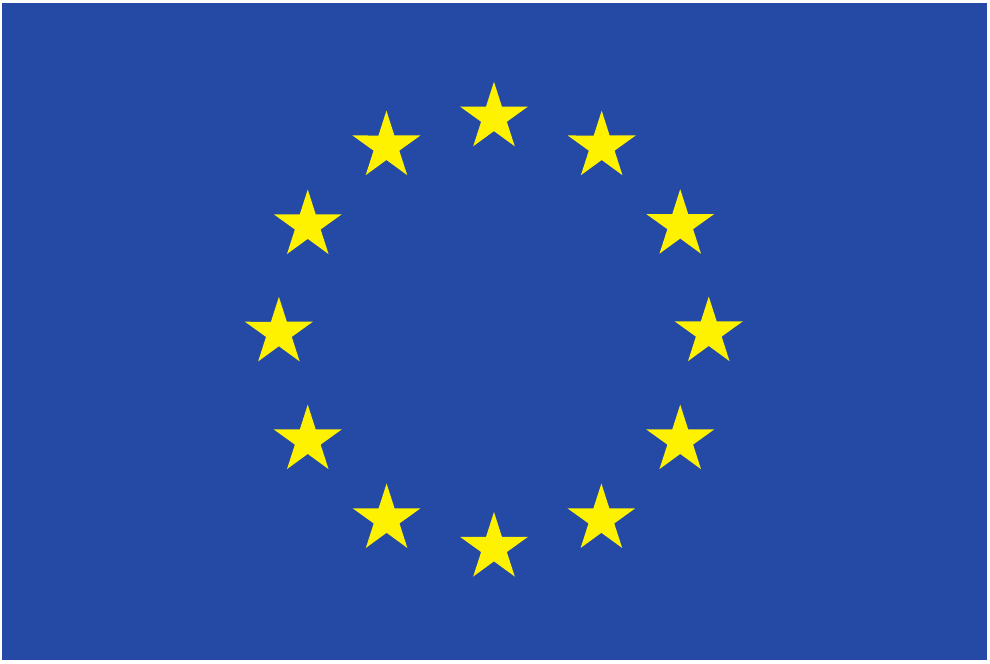}\flag{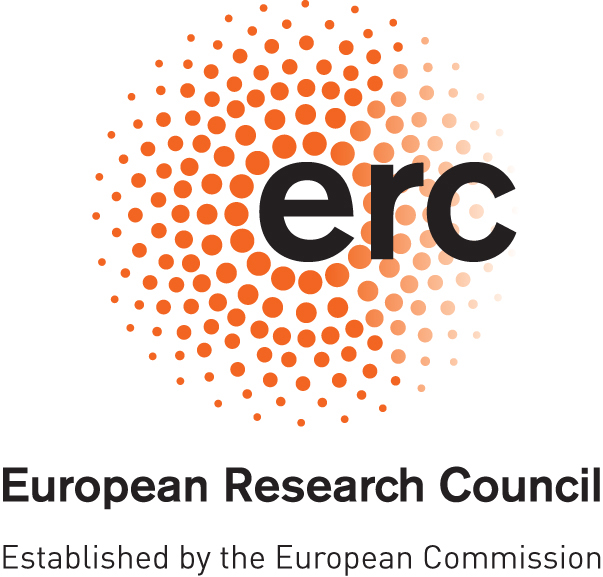}This work is part of the project CUTACOMBS that has received funding from the European Research Council (ERC) under the European Unions Horizon 2020 research and innovation programme (grant agreement No. 714704).}

\authorrunning{J. Blum, Y. Disser, A. E. Feldmann, S. Gupta, A. Zych-Pawlewicz}

\Copyright{Johannes Blum, Yann Disser, Andreas Emil Feldmann, Siddharth Gupta,
and Anna Zych-Pawlewicz}

\ccsdesc[500]{Theory of computation~Fixed parameter tractability}
\ccsdesc[500]{Theory of computation~Approximation algorithms analysis}

\keywords{
sparse hitting set,
fair vertex cover,
highway dimension
}%
\EventEditors{}
\EventNoEds{0}
\EventLongTitle{}
\EventShortTitle{}
\EventAcronym{}
\EventYear{2022}
\EventDate{}
\EventLocation{}
\EventLogo{}
\SeriesVolume{}
\ArticleNo{}

\begin{document}

\maketitle

\begin{abstract}
We consider the \pname{Sparse Hitting Set (\shs)} problem, where we are given a 
set system~$(V,\mc{F},\mc{B})$ with two families $\mc{F},\mc{B}$ of subsets of 
the universe $V$. The task is to find a hitting set for~$\mc{F}$ that minimizes 
the maximum number of elements in any of the sets of~$\mc{B}$. This generalizes 
several problems that have been studied in the literature. Our focus is on 
determining the complexity of some of these special cases of \shs with respect 
to the \emph{sparseness}~$k$, which is the optimum number of hitting set 
elements in any set of $\mc{B}$ (i.e., the value of the objective function).

For the \pname{Sparse Vertex Cover (\svc)} problem, the universe is given by 
the vertex set~$V$ of a graph, and $\mc{F}$ is its edge set. We prove 
NP-hardness for sparseness $k\geq 2$ and polynomial time solvability for $k=1$. 
We also provide a polynomial-time $2$-approximation algorithm for any $k$. A
special case of \svc is \pname{Fair Vertex Cover (\fvc)}, where the family
$\mc{B}$ is given by vertex neighbourhoods. For this problem it was open
whether it is FPT (or even XP) parameterized by the sparseness~$k$. We answer
this question in the negative, by proving NP-hardness for constant~$k$. We
also provide a polynomial-time $(2-\frac{1}{k})$-approximation algorithm for
\fvc, which is better than any approximation algorithm possible for \svc or
the \pname{Vertex Cover} problem (under the Unique Games Conjecture).

We then switch to a different set of problems derived from \shs related to the 
\emph{highway dimension}, which is a graph parameter modelling transportation 
networks. In recent years a growing literature has shown interesting algorithms 
for graphs of low highway dimension. To exploit the structure of such graphs, 
most of them compute solutions to the \pname{$r$-Shortest Path Cover (\spc)} 
problem, where $r>0$, $\mc{F}$ contains all shortest paths of length between 
$r$ and $2r$, and $\mc{B}$ contains all balls of radius $2r$. It is known that 
there is an XP algorithm that computes solutions to \spc of sparseness at most
$h$ if the input graph has highway dimension $h$. However it was not known 
whether a corresponding FPT algorithm exists as well. We prove that \spc and 
also the related \pname{$r$-Highway Dimension (\hd)} problem, which can be used 
to formally define the highway dimension of a graph, are both W[1]-hard. 
Furthermore, by the result of Abraham et al.~[ICALP~2011] there is a 
polynomial-time $O(\log k)$-approximation algorithm for \hd, but for \spc such 
an algorithm is not known. We prove that \spc admits a polynomial-time $O(\log 
n)$-approximation algorithm.

\end{abstract}

\section{Introduction}

In this paper, we study the problem of finding a sparse hitting set. That is, 
we are given a set system~$(V,\mc{F},\mc{B})$ on universe $V$ with two set 
families $\mc{F},\mc{B}\subseteq 2^V$, and a feasible solution is a set 
$H\subseteq V$ that hits (i.e.,~intersects) every set of~$\mc{F}$. Instead of 
minimizing the overall size of the solution however, we think of the sets of 
$\mc{B}$ as being small and we would like to distribute the solution $H$ among 
the sets in $\mc{B}$ as evenly as possible. Intuitively and depending on the 
context, the sets in $\mc{B}$ are balls in some metric and the hitting set 
should be sparse within them. That is, we want to find a hitting set for 
$\mc{F}$ that minimizes the largest intersection with the sets of $\mc{B}$. 
Formally, the \pname{Sparse Hitting Set} (\shs) problem with 
input~$(V,\mc{F},\mc{B})$ is defined by the following integer linear program 
(ILP) with indicator variables $x_v$ for each $v\in V$ encoding membership in 
the solution $H\subseteq V$.
\begin{align*}
	\min k \text{ such that: }\qquad 
	\sum_{v\in F} x_v &\geq 1 &\forall F &\in\mc{F} 
	\tag{\shs-ILP}\label{shs-ILP}\\
	\sum_{v\in B} x_v &\leq k &\forall B &\in\mc{B}\\
	x_v &\in\{0,1\} &\forall v &\in V
\end{align*}

The \shs problem generalizes several problems studied in the literature, with 
applications in for instance cellular~\cite{DBLP:conf/cocoon/KuhnRWWZ05}, 
communication~\cite{lin1989fair}, and 
transportation~\cite{abraham2010highway,knop2019parameterized} networks. Our aim 
in this paper is to determine the complexity of some basic variants of \shs, and 
we are specifically interested in the complexity depending on the 
\emph{sparseness}, which is the solution value $k$ of \eqref{shs-ILP}.
In general, \shs contains the \pname{Hitting Set} problem by 
setting~$\mc{B}=\{V\}$, and thus does not admit any $g(k)$-approximation in 
$f(k)\polyn$ time~\cite{karthik2019parameterized}, for any computable functions 
$f$ and~$g$, where $n=|V|$, under ETH. 

\subparagraph*{Sparse Vertex Cover.} A much easier special case of 
\pname{Hitting Set} is the well-known \pname{Vertex Cover} problem: for the 
\pname{Sparse Vertex Cover (\svc)} problem the set system is given by a graph 
$G=(V,E)$ so that $\mc{F}=E$ and~$\mc{B}\subseteq 2^V$.
We show that this problem is NP-hard for any $k \geq 2$, even on very simple
input graphs.

\begin{theorem}[restate=svcHard,name=]\label{thm:svc-hard}
\svc is \textup{NP}-hard for any $k\geq 2$, even if the input graph is a 
matching. 
\end{theorem}

Note that this hardness result implies that, unless P=NP, \svc does not admit an 
\emph{XP algorithm} with runtime~$n^{f(k)}$ for any function $f$ (the problem is 
paraNP-hard parameterized by the sparseness~$k$). This is in contrast to the 
\pname{Vertex Cover} problem, which is known to be \emph{fixed-parameter 
tractable (FPT)} parameterized by the solution size~$s$, which means that it 
can be solved much more efficiently in~$f(s)\polyn$ time for some function~$f$ 
(which can be shown~\cite{DBLP:conf/mfcs/ChenKX06} to be $1.2738^s$). On the other 
hand, we will show that for $k=1$ the \svc problem is polynomial-time solvable, 
which together with the previous hardness result settles the complexity of the 
\svc problem for every sparseness value~$k$.

\begin{theorem}[restate=svcPoly,name=]\label{thm:svc-poly}
\svc is polynomial time solvable for $k=1$. 
\end{theorem}

As \pname{Vertex Cover} is a special case of \svc with $\mc{B} = 
\{V\}$, any polynomial time $(2-\eps)$\hy{}approximation algorithm for \svc 
would refute the Unique Games Conjecture~(UGC)~\cite{williamson2011design}. On 
the positive side, we show that we can match this conditional approximation 
lower bound with a $2$-approximation algorithm. This means that \mbox{\svc} can 
be approximated as well as the \pname{Vertex Cover} problem, which also admits 
a $2$\hy{}approximation~\cite{williamson2011design} in polynomial time.

\begin{theorem}[restate=svcApprox,name=]\label{thm:svc-approx}
\svc admits a polynomial time $2$-approximation algorithm. 
\end{theorem}

\subparagraph*{Fair Vertex Cover.}
A special case of \svc is the \pname{Fair Vertex Cover (\fvc)} problem where the 
family of sets~$\mc{B}$ is given by closed neighbourhoods, i.e., if $N[v]$ is 
the set containing vertex~$v$ and all neighbours of $v$ in $G$ then 
$\mc{B}=\{N[v]\mid v\in V\}$ (alternatively, $\mc{B}$ contains all balls of 
radius~$1$). The fairness constraint was introduced by \citet{lin1989fair} in 
the context of communication networks, and has since then been studied for 
several types of problems (cf.~\cref{sec:related}), including \pname{Vertex 
Cover}~\cite{knop2019parameterized,jacob2019deconstructing, 
masavrik2020parameterized}. In contrast to this paper, 
in~\cite{knop2019parameterized,jacob2019deconstructing} the problem is defined 
slightly differently by considering open neighbourhoods, i.e., 
$\mc{B}=\{N[v]\setminus\{v\}\mid v\in V\}$, and we call this version \ofvc. 
Notably, the parameterized complexity of \ofvc has been studied for a plethora 
of parameters, including treedepth, treewidth, feedback vertex set, modular 
width \cite{knop2019parameterized}, and the total solution size $|H|$
\cite{jacob2019deconstructing}, and most of these results also apply to 
\fvc with closed neighbourhoods. 

\citet{jacob2019deconstructing} observe that it is NP-hard to decide if a vertex 
cover of size $s$ exists, if every  neighbourhood is allowed to only contain at 
most $k$ vertices of the solution~$H$, for a given constant $k\geq 3$: this 
follows from the fact that \pname{Vertex Cover} is NP-hard on sub-cubic 
graphs~\cite{garey1974some}. While the authors of~\cite{jacob2019deconstructing} 
call this problem \pname{Fair Vertex Cover} as well, note that this is 
significantly different from the \fvc problem studied in this paper as well as 
the \ofvc problem studied in~\cite{knop2019parameterized}.
In particular, on sub-cubic graphs both of these problems as defined here 
always trivially have a solution for $k\geq 3$,\footnote{Observe that it is never
necessary to pick a vertex $v$ and all its neighbours.} and thus the NP-hardness 
of \pname{Vertex Cover} on sub-cubic graphs does not immediately imply 
NP-hardness of \fvc or \ofvc. In fact, for the natural parameterization by the 
sparseness $k$ the complexity of \ofvc (and also \fvc) has so far been 
unknown.\footnote{Tomáš Masařík, personal communication.} We answer this open 
problem by showing NP-hardness of \fvc for $k\geq 3$ and of \ofvc for $k\geq 4$ 
on more complex input graphs when compared to \svc.

\begin{theorem}[restate=fvcHard,name=]\label{thm:fvc-hard}
\fvc is \textup{NP}-hard for any $k\geq 3$ and \ofvc is \textup{NP}-hard for 
any~$k\geq 4$, even on planar input graphs.
\end{theorem}

Thus, as for \svc, we can conclude that \fvc and \ofvc do not admit XP 
algorithms parameterized by~$k$. 
For the cases when $k\leq 2$, \citet{jacob2019deconstructing} provide a 
polynomial time algorithm that solves their version of \pname{Fair Vertex 
Cover}, which however also works for the \fvc and \ofvc problems as defined in 
this paper.
Hence this settles the complexity of \fvc for every value of $k$, and only 
leaves the value $k=3$ open for \ofvc.

In terms of approximation, interestingly we are able to obtain a slightly better 
algorithm for \fvc than for \svc, namely a $(2-\frac{1}{k})$-approximation. 
This beats the best possible approximation for \svc and \pname{Vertex Cover} 
under UGC~\cite{williamson2011design}. In particular, the following theorem 
implies that for the smallest value $k=3$ for which \fvc is NP-hard, we can 
obtain a solution of sparseness $5$. We leave open whether a solution of 
sparseness $4$ can be computed in polynomial time for \fvc if~$k=3$, and whether 
better approximation algorithms are possible for \ofvc.

\begin{theorem}[restate=fvcApprox,name=]\label{thm:fvc-approx}
\fvc admits a polynomial time $(2-\frac{1}{k})$-approximation algorithm. 
\end{theorem}

\subparagraph*{Shortest Path Cover and Highway Dimension.}
We now turn to a different set of problems derived from \shs, which as we shall 
see generalize \fvc. Given a value $r>0$ and an edge-weighted graph $G$, for 
the \pname{$r$-Shortest Path Cover (\spc)} problem the family $\mc{F}$ is given 
by shortest paths of length between $r$ and $2r$ and the family $\mc{B}$ is 
given by balls of radius~$2r$. That is, let $\mc{P}_r$ contain $S\subseteq V$ 
if and only if $S$ is the vertex set of a path in $G$, which is a shortest path 
(according to the edge weights) and whose length is in the range $(r,2r]$.
Furthermore, let $\dist(u,v)$ be the length of a shortest $u$-$v$-path and let
$B_r(v)=\{u\in V\mid\dist(u,v)\leq r\}$ denote the ball of radius $r$ centered 
at $v$. Then for the \spc problem, $\mc{F}=\mc{P}_r$ and $\mc{B}=\{B_{2r}(v)\mid 
v\in V\}$. 

The \spc problem finds applications in the context of the \emph{highway 
dimension}, which is a graph parameter introduced by \citet{abraham2010highway} 
to model transportation networks. To define the highway dimension, we define a 
problem related to \spc called \pname{$r$-Highway Dimension~(\hd)}, where for 
each vertex $v\in V$ the task is to find a hitting set for all shortest paths of 
length in~$(r,2r]$ intersecting the ball $B_{2r}(v)$, and we need to minimize 
the largest such hitting set. Note that compared to \spc the quantification is 
reversed, i.e., for \spc there is a hitting set that is small in every ball, 
while for \hd for every ball there is a small hitting set (thus \hd is not a 
special case of \shs). The \emph{highway dimension} of an edge-weighted graph 
$G$ is the smallest integer $h$ such that there is a solution to \hd of value at 
most $h$ in $G$ for every $r>0$.

There is empirical evidence~\cite{bast2007transit} that road networks have small 
highway dimension, and it has been conjectured~\cite{feldmann20181+varepsilon} 
that public transportation networks (especially those stemming from airplane 
networks) have small highway dimension as well.\footnote{In fact there are 
several definitions of the highway dimension, with the one presented here being 
well-suited for public transportation networks, 
cf.~\cite{feldmann20181+varepsilon,blum2019hierarchy}} Therefore, there has been 
some effort to devise algorithms~\cite{feldmann20181+varepsilon, 
abraham2010highway, DBLP:journals/algorithmica/Feldmann19, 
DBLP:journals/jcss/FeldmannS21, DBLP:journals/algorithmica/DisserFKK21, 
DBLP:journals/algorithmica/DisserFKK21, becker2018polynomial, 
jayaprakash2021approximation, braverman2021coresets, bohm2022hop} for problems 
on low highway dimension graphs that naturally arise in transportation networks. 
It is known~\cite{abraham2011vc} that if the highway dimension of a graph $G$ is 
$h$, then the \spc problem on $G$ has sparseness at most~$h$ for every~$r>0$, 
but not vice versa, as the sparseness of \spc can be much smaller than~$h$. 
However, since a solution to \spc consists of one hitting set $H\subseteq V$ for 
the whole graph, it is more convenient to work with algorithmically than the 
$n$ hitting sets for all balls of radius~$2r$ that form a solution to \hd. 
Therefore, algorithms exploiting the structure of graphs of low highway 
dimension typically compute a solution to the \spc problem for each of 
the~$O(n^2)$ relevant values of $r$ given by the pairwise distances between 
vertices. 

For graphs of low highway dimension, \citet{abraham2011vc} give an algorithm 
that for each relevant value of $r$ computes a solution to the \hd problem, in 
order to obtain a solution to \spc with sparseness at most the value of the \hd 
solution. While \citet{abraham2011vc} propose to use an approximation algorithm 
for \hd (see below), note that the \hd problem admits an XP algorithm with 
runtime $n^{O(k)}$, since for any ball~$B_{2r}(v)$ it can construct the set 
system given by all shortest paths of length in $(r,2r]$ intersecting 
$B_{2r}(v)$, for which it can then try every possible $k$-tuple of vertices as a 
solution. This algorithm can thus be used to compute solutions to \spc of 
sparseness at most~$h$ in $n^{O(h)}$ time if the input graph has highway 
dimension $h$. Interestingly, it is not possible to compute solutions of optimum 
sparseness for \spc using an XP algorithm due to the NP-hardness of \fvc: 
consider an \spc instance with unit edge weights and value $r=1/2$. Since every 
edge is a shortest path between its endpoints, the \spc problem on this instance 
is equivalent to \fvc. As argued above however, no XP algorithm exists for \fvc, 
unless P=NP.

In light of the growing amount of work on problems on low highway dimension 
graphs, it would be very useful to have a faster algorithm to solve \hd in order 
to compute a hitting set for \spc of corresponding sparseness. While it is known 
that computing the highway dimension is NP-hard~\cite{feldmann20181+varepsilon} 
and this also implies that \hd is NP-hard, \hd might still be FPT and allow 
algorithms with runtime $f(k)\polyn$ for some function $f$. However, we will
show that it is unlikely that such algorithms exist. In particular, we prove
that \hd is \textup{W[1]}-hard parameterized by the solution value $k$.
We also prove that 
\spc does not admit FPT algorithms (in particular, $k$ here denotes the optimum 
sparseness and not just an upper bound that we would obtain by solving \hd, as 
suggested above). While already the above reduction from \fvc to \spc excludes 
FPT algorithms for \spc, this only excludes such algorithms for very small 
values of $r$, in fact the smallest relevant value for~$r$ (as the problem 
becomes trivial for even smaller values). A priori it is not clear whether \spc 
admits FPT (or XP) algorithms for large values of $r$. In our reduction 
 however, the value of $r$ takes the largest relevant value, so 
that there exists a ball of radius $2r$ containing the whole graph.

\begin{theorem}[restate=hdHard,name=]\label{thm:hd-hard}
Both \hd and \spc are \textup{W[1]}-hard parameterized by their solution 
values~$k$, where $2r$ is the radius of the input graph. 
\end{theorem}

One caveat of this hardness result is that it does \emph{not} answer the 
question of whether computing the highway dimension is FPT or not. This is 
because the presented reduction only shows hardness of \hd for a large value 
$r$. However, for smaller values of~$r$ the solution value to \hd is unbounded 
in the constructed graph, and thus the graph does not have bounded highway 
dimension. This means that it might still be possible to compute the highway 
dimension in FPT time, but not using the existing tools provided by 
\citet{abraham2011vc}, where each value $r$ is considered separately. Instead, 
if such an algorithm exists it must consider the structure of the whole graph. 
We leave open whether there is such an algorithm.

As mentioned above, \citet{abraham2011vc} propose an approximation algorithm for 
\hd: under the assumption that all shortest paths are unique (which can always 
be achieved by slightly perturbing the edge lengths), \hd admits a polynomial 
time $O(\log k)$-approximation algorithm. Due to the fact that the sparseness of 
\spc can be a lot smaller than the solution value to~\hd,\footnote{as for 
instance witnessed by the graphs constructed in the reduction for 
\cref{thm:fvc-hard} and value $r=1/2$.} it is not known how to obtain such an 
algorithm for \spc. However, we prove the existence of a weaker $O(\log
n)$-approximation algorithm.

\begin{theorem}[restate=hdApprox,name=]\label{thm:hd-approx}
\spc admits a polynomial time $O(\log n)$-approximation algorithm. 
\end{theorem}

\subparagraph*{Dense Matching.}
Finally, in light of the above results for \svc, we also 
consider the dual \pname{Dense Matching} problem, where we are given a graph 
$G=(V,E)$ and the task is to find a matching $M\subseteq E$ maximizing the 
smallest number of matching edges induced by a set in the family $\mc{B}$, i.e., 
the minimum $|M\cap E(B)|$ over all~$B\in\mc{B}$, where $E(B)=\{\{u,v\}\in E\mid 
u,v\in B\}$. Despite the \pname{Maximum Matching} problem being polynomial-time 
solvable, we show that \pname{Dense Matching} does not admit a polynomial time 
$(2-\eps)$-approximation, even if~$\mc{B}$ is restricted to balls of radius two, 
unless P=NP. Interestingly, a matching $2$-approximation seems a lot harder to 
come by compared to \svc, and we leave open whether a constant approximation is 
possible for \pname{Dense Matching}.

\begin{theorem}[restate=matchingHard,name=]\label{thm:matching-hard}
It is NP-hard to approximate \pname{Dense Matching} within $2-\eps$ for any 
$\eps>0$, even if $\mc{B}=\{B_2(v)\mid v\in V\}$ where all edges have 
weight~$1$.
\end{theorem}

\subsection{Related Work}\label{sec:related}
Apart from the work cited above, we here list some additional related work. 
\citet{kanesh2021parameterized} study the \pname{Fair Feedback Vertex Set} 
problem, where the family $\mc{F}$ contains all vertex sets of cycles of the 
input graph (in this case $\mc{F}$ is not part of the input). They prove 
results on the parameterized complexity of several versions of this problem, 
where the considered parameters include treewidth, treedepth, neighbourhood 
diversity, the total solution size, and the maximum vertex degree. 
\citet{jacob2019deconstructing} consider the parameterized complexity of the 
\pname{Fair Set} and \pname{Fair Independent Set} problems, but also 
\pname{$\Pi$-Fair Vertex Deletion}, where $\Pi$ is any property expressible in 
first order (FO) logic. \citet{knop2019parameterized} study 
\pname{$\Pi$\hy{}Fair Vertex Deletion} for properties $\Pi$ expressible in 
monadic second order ($\text{MSO}_1$) logic parameterized by the twin cover 
number. They also consider \fvc parameterized by treedepth, feedback vertex 
number, and modular width. \citet{agrawal2022parameterized} study the 
parameterized complexity of the \pname{Minimum Membership Dominating Set} 
problem, where $\mc{F}=\mc{B}=\{N[v]\mid v\in V\}$, and consider 
parameterizations by pathwidth, sparseness, and vertex cover number.

While in this paper we study \shs problems on graphs where the universe is the 
set of vertices, another line of work studies variants of \shs when the universe 
is the edge set. For instance, the work of \citet{lin1989fair} that introduced 
the fairness constraint, studies the \pname{Fair Feedback Edge Set} problem, 
where the family~$\mc{F}$ contains the edge sets of all cycles of the input 
graph. \citet{masavrik2020parameterized} consider the parameterized complexity 
of the \pname{$\Pi$-Fair Edge Deletion} problem, where $\Pi$ is a property 
expressible in FO logic or in MSO logic. For each of these problems they study 
parameterizations by the treewidth, pathwidth, treedepth, feedback vertex set 
number, neighbourhood diversity, and vertex cover number. \citet{kolmanfair} 
study the \pname{$\Pi$-Fair Edge Deletion} problem on graphs of bounded 
treewidth, where $\Pi$ is any property expressible in MSO logic. They also give 
tight polynomial-time $O(\sqrt{n})$-approximation algorithms for \pname{Fair Odd 
Cycle Transversal} and \pname{Fair Min Cut}, where the family $\mc{F}$ contains 
all edge sets of odd cycles and $(s,t)$-paths for given vertices $s$ and~$t$, 
respectively. Another notable problem is \pname{Min Degree Spanning Tree}, where 
the family $\mc{F}$ consists of every edge cut under the fairness constraint. 
\citet{furer1992approximating} prove that the problem is NP-hard but a solution 
of sparseness $k+1$ can be computed in polynomial-time.

Regarding the complexity of computing the highway dimension, it is interesting 
to note that \citet{abraham2011vc} show that any set system given by unique 
shortest paths has VC\hy{}dimension~$2$ (this observation also leads to the 
above mentioned $O(\log k)$-approximation algorithm for \hd). At the same 
time, \citet{bringmann2016hitting} prove that the \pname{Hitting Set} problem 
is W[1]-hard for set systems of VC-dimension~$2$. Hence it is intriguing to 
think that the latter reduction could possibly be modified to also prove 
W[1]-hardness for \hd or \spc. However, it seems that shortest paths exhibit a 
lot more structure than general set systems of VC-dimension~$2$, and thus it is 
unclear how to obtain a hardness result for \hd or \spc based 
on~\cite{bringmann2016hitting}. Instead, a more careful reduction as provided in 
\cref{thm:hd-hard} seems necessary.

\section{Sparse Vertex Cover}\label{sec:svc}
In this section we consider the \svc problem and start by proving NP-hardness for any
$k \geq 2$.

\svcHard*
\begin{proof}
We reduce from a variant of the satisfiability problem
called \pname{exactly-3-Sat}, meaning that all clauses contain
exactly three literals. This problem was shown to be NP-complete 
in~\cite{GareyJohnson/79}.
For a set of variables $X=\{x_{i}\}_{i}$, we use the notation 
$\bar{X}:=\{\bar{x}_{i}\}_{i}$.

Let an instance of \pname{exactly-3-Sat}
be given by a set of variables $X=\{x_{i}\}_{i=1,\dots,n}$
and a set of clauses $\mathcal{C}=\{C_{j}\}_{j=1,\dots,m}$ with $C_{j}\subset 
X\cup\bar{X},|C_{j}|=3$.
We define the graph $G=(V,E)$ by $V=X\cup\bar{X}\cup
\{y_{i},\bar{y}_{i} \mid i \in \{1, \dots, k-1\}\}$
and $E=\{\{x_{i},\bar{x}_{i}\} \mid i\in\{1,\dots,n\}\} \cup
\{\{y_{i},\bar{y}_{i}\} \mid i\in\{1,\dots,k-1\}\}$. We further let
$\tilde{C} = C \cup \{y_{1},\bar{y}_{1},\dots, y_{k-2},\bar{y}_{k-2}\}$
and choose
$\mc{B}=\{\{x_{i},\bar{x}_{i},y_{1},\bar{y}_{1},\dots,
y_{k-1},\bar{y}_{k-1}\} \mid i\in\{1,\dots,n\}\} \cup
\{ \tilde{C} \mid C \in \mc{C} \}$.
This construction can be carried out in linear time and~$G$ is a
matching. For NP-hardness, it remains to show that $G$ has a vertex cover
$H \subseteq V$ satisfying $|H \cap B| \leq k$ for every $B \in \mc{B}$
if and only if the given $\pname{exactly-3-Sat}$ instance has a
satisfying assignment. 

To see this, first assume that the given $\pname{exactly-3-Sat}$
instance has a satisfying assignment $\alpha\colon X\to\{0,1\}$ and
extend $\alpha$ to $\bar{X}$ by letting $\alpha(\bar{x}):=1-\alpha(x)$.
We construct the vertex cover $H=\{x\in X\cup\bar{X} \mid \alpha(x)=0\}\cup
\{y_{1}, \dots, y_{k-1}\}$.
Indeed, $H$ is a vertex cover, since every edge of $G$ is covered by exactly
one of its endpoints.
It also follows that, for every $i\in\{1,\dots,n\}$, we have 
$\left|\{x_{i},\bar{x}_{i},y_{1},\bar{y}_{1},\dots,y_{k-1},\bar{y}_{k-1}\}
\cap H\right|=k$.
By definition of $\alpha$, for every $C\in\mathcal{C}$, we have
$\sum_{x\in C}\alpha(x)\geq1$, hence $\left|H\cap \tilde{C}\right|=
\left|H \cap C\right| + k-2 = \sum_{x\in C}\alpha(\bar{x}) + k-2=
3-\sum_{x\in C}\alpha(x) + k-2 \leq k$.

Conversely, suppose that there exists a vertex
cover $H\subseteq V$ with $\left|H \cap B\right|\leq k$ for all $B\in\mc{B}$.
We claim that $\alpha(x)=\left|\{\bar{x}\}\cap H\right|$ defines
a satisfying assignment for the given $\pname{exactly-3-Sat}$ instance.
Observe that we must have $\left|\{x_{i},\bar{x}_{i}\}\cap H\right|\geq 1$
for all $i\in\{1,\dots,n\}$ and $\left|\{y_{i},\bar{y}_{i}\}\cap H\right|\geq 1$
for all $i\in\{1,\dots,k-1\}$, since $H$ needs to cover all edges.
Since $\{x_{i},\bar{x}_{i},y_{1},\bar{y}_{1},\dots,y_{k-1},\bar{y}_{k-1}\}
\in\mc{B}$, it
follows that $\left|\{x_{i},\bar{x}_{i},y_{1},\bar{y}_{1},\dots,
y_{k-1},\bar{y}_{k-1}\}\cap H\right|\leq k$
for all $i\in\{1,\dots,n\}$. Together, we obtain 
$\left|\{x_{i},\bar{x}_{i}\}\cap H\right|=1$
for all $i\in\{1,\dots,n\}$. We can therefore extend $\alpha$ to
$\bar{X}$ by setting $\alpha(\bar{x})=1-\alpha(x)=\left|\{x\}\cap H\right|$.
Finally, for $C\in\mc{C}$, we have $\tilde C\in\mc{B}$ and thus
$\left|H\cap \tilde C\right|\leq k$ and moreover
$\left| H \cap \tilde{C} \right| \geq \left| H \cap C \right| + k-2$, which
implies $\left| H \cap C \right| \leq 2$.
It follows that $\sum_{x\in 
C}\alpha(x)=3-\sum_{x\in C}\alpha(\bar{x})=3-\sum_{x\in C}\left|\{x\}\cap 
U\right|=3-\left|H\cap C\right|\geq1$,
thus $\alpha$ is a satisfying assignment.
\end{proof}

We can observe that \cref{thm:svc-hard} also shows that \svc does not admit
a $(3/2 - \eps)$-approximation algorithm for any $\eps>0$, unless P=NP. This
follows from the fact that for~$k=2$, such an algorithm would be able to
determine whether a given instance of \svc admits a solution of sparseness
$2 \cdot (3/2 - \eps) < 3$, i.e., of optimal sparseness $2$.

Let us now consider the \svc problem for sparseness $k=1$. We show that in this
case, \svc can be reduced to the $\pname{2-SAT}$ problem, which is commonly
known to admit a linear time algorithm. This yields the following theorem.

\svcPoly*
\begin{proof}
The instance of the \svc problem is given by a graph $G=(V,E)$ and a set of 
balls~$\mc{B}\subseteq 2^V$. Given this instance, we construct a 
$\pname{2-Sat}$ formula $\phi$, which is solvable if and only if the \svc 
instance has a solution. Moreover, we can reconstruct the fair vertex cover for 
$(V,E,\mc{B})$ given a satisfying assignment to $\phi$.

To construct $\phi$, we first assign a variable $x_v$ to each vertex $v \in V$. 
Next, for every edge $\{u,v\} \in E$ we create a clause $(x_u \vee x_v)$ and add it 
to $\phi$, so that we are guaranteed that any satisfying assignment will 
correspond to a valid vertex cover. Now we have to enforce, that for each ball 
$B \in \mc{B}$, at most one variable in the set $\{ x_v \}_{v \in \mc{B}}$ is 
set to true. This is done by adding ${|B|}\choose{2}$ clauses: for each pair 
$v,u \in B, u \neq v$ we add a clause $(\bar{x}_v \vee \bar{x}_u)$ to enforce 
that $x_v$ and $x_u$ cannot be both true. In this way, we ensure that only one 
variable of the set $\{ x_v \}_{v \in \mc{B}}$ is set to true. Thus, the final 
formula $\phi$ takes the following form:
\[
\phi=\bigwedge_{\{u,v\} \in E} (x_u \vee x_v) \wedge \bigwedge_{B \in \mc{B}, 
u,v \in B: u \neq v} (\bar{x}_v \vee \bar{x}_u)
\]

Given a satisfying assignment for $\phi$, we reconstruct the solution to 
$(V,E,\mc{B})$ by taking the vertices whose variable was set to true. We already 
argued that such a solution is feasible for \svc with $k=1$. In the opposite 
direction, given a solution to \svc with~$k=1$, we find an assignment by setting 
the variables corresponding to the vertices of the solution to true: the clauses 
corresponding to edges are then satisfied due to the solution being a vertex 
cover, and the remaining clauses corresponding to $\mc{B}$ are satisfied because 
the solution picks at most one vertex from each $B \in \mc{B}$.
\end{proof}

Finally, we show how to obtain a $2$-approximation algorithm for \svc.
This approximation factor is optimal unless the Unique Games Conjecture fails,
as \pname{Vertex Cover} is a special case of \svc with $\mc{B}=\{V\}$.

\svcApprox*
\begin{proof}
We consider the relaxation of \eqref{shs-ILP} for a given graph $G=(V,E)$:
\begin{align}
 \min k \text{ such that: }\qquad 
 x_u+x_v &\geq 1 &\forall uv &\in E \label{eq:VC}\\
 \sum_{v\in B} x_v &\leq k &\forall B &\in\mc{B} \label{eq:sparse}\\
 x_v &\geq 0 &\forall v &\in V \label{eq:nneg}
\end{align}

Note that in any feasible solution to this LP, for any edge $\{u,v\}$ at least one 
of the two variables $x_u$ and $x_v$ has value at least $1/2$, due to 
constraints~\eqref{eq:VC} and~\eqref{eq:nneg}. Thus the set $W=\{v\in V\mid 
x_v\geq 1/2\}$ of all vertices with value at least $1/2$, is a vertex cover for 
the input graph. The sparseness of this solution can be bounded 
using~\eqref{eq:sparse} for any set $B\in\mc{B}$:
\[
|W\cap B|\leq 2 \sum_{v\in B} x_v \leq 2k
\]
Thus solving the above LP relaxation optimally in polynomial time and then 
outputting the set $W$, gives a $2$-approximation algorithm for \svc.
\end{proof}

\section{Fair Vertex Cover}
Let us now consider the \pname{(Open-)Fair-VC} problem, where the balls
$\mc{B}$ are given by (open) vertex neighborhoods.
We first show NP-hardness of \fvc and \ofvc for $k \geq 3$ and $k \geq 4$,
respectively.

\fvcHard*
\begin{proof}
We reduce from the \pname{planar 2P1N-3-Sat} problem.
In this variant of satisfiability, all clauses contain two or three
literals, and we may assume that every variable appears exactly twice
as a positive literal and exactly once as a negative literal over
all clauses. In addition, we may assume that the bipartite graph connecting
clauses to the variables they contain is planar. This variant of satisfiability
was shown to be NP-complete in~\cite{ManuchGaur/08}.

We first consider the \fvc problem and later show how to modify our reduction
for \ofvc.
Let an instance of \pname{planar 2P1N-3-Sat} be given by a set of variables 
$X=\{x_{i}\}_{i=1,\dots,n}$ and a set of clauses 
$\mathcal{C}=\{C_{j}\}_{j=1,\dots,m}$ with $C_{j}\subset X\cup\bar{X}$, $|C_j| 
\in \{2,3\}$.
We define the graph $G=(V,E)$ by 
\[
	V = \bigcup_{i=1}^{n}\left(\{x_{i},\bar{x}_{i}\} \cup \bigcup_{s=1}^{k-2} \bigcup_{r=0}^k \{y_{i,r}^s\} \right) \cup
	    \bigcup_{j=1}^{m}\left(\bigcup_{r=0}^k \{z_{j,r}\} \cup \bigcup_{s=1}^{k-|C_j|} \bigcup_{r=0}^{k} \{q_{j,r}^{s} \} \right)
\]
 and 
\begin{align*}
  E & = \bigcup_{i=1}^{n}\left(\{\{x_i,\bar{x}_i\}\} \cup \bigcup_{s=1}^{k-3} \{\{x_i,y_{i,0}^s\}\} \cup \bigcup_{s=1}^{k-2} \{\{\bar{x}_i,y_{i,0}^s\}\} \cup \bigcup_{s=1}^{k-2} \bigcup_{r=1}^k \{\{y_{i,0}^s,y_{i,r}^s\}\} \right) \\
    & \cup\bigcup_{j=1}^{m}\left(\bigcup_{x\in C_{j}}\{\{x,z_{j,0}\}\} \cup
	\bigcup_{r=1}^k\{\{z_{j,0},z_{j,r}\}\} \cup \bigcup_{s=1}^{k-|C_j|} \left( \{z_{j,0},q_{j,0}^{s}\} \cup \bigcup_{r=1}^{k} \{q_{j,0}^{s},q_{j,r}^{s}\} \right) \right). \\
\end{align*}
This construction (illustrated in \cref{fig:fvc}) can be carried out in linear time and $G$ is planar.

\begin{figure}[t]
\centering
\includegraphics{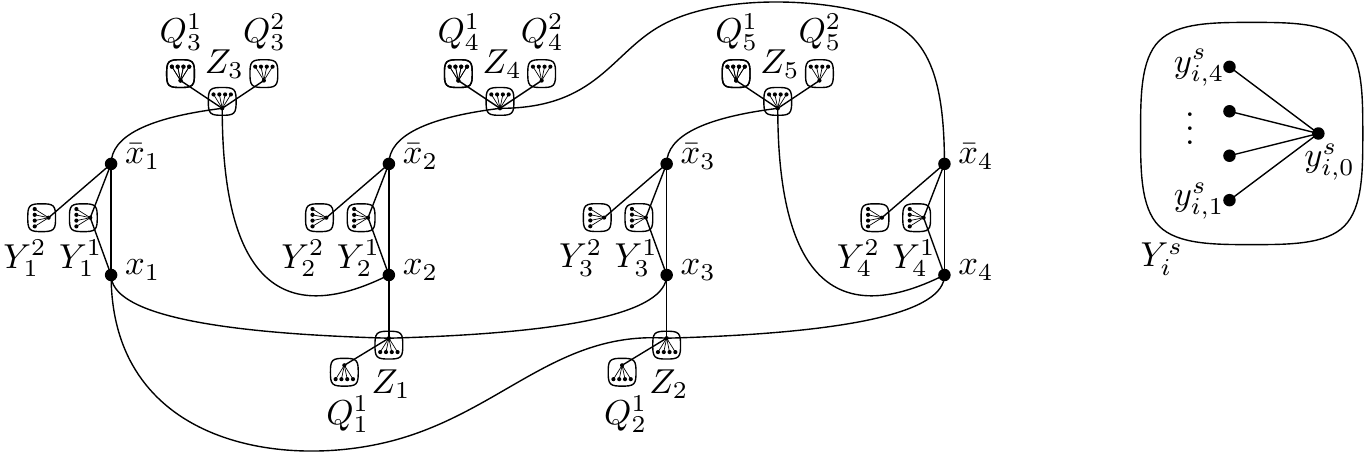}
\caption{Left: The graph $G$ for the formula $(x_1 \vee x_2 \vee x_3) \wedge (x_1 \vee x_3 \vee x_4) \wedge (\bar{x}_1 \vee x_2) \wedge (\bar{x}_2 \vee \bar{x}_4) \wedge (\bar{x}_3 \vee x_4)$ and $k=4$. Right: A star $Y_i^s$ with center $y_{i,0}^s$ and leaves $y_{i,1}^s,\dots,y_{i,4}^s$. Similarly, $Z_j$ and $Q_{j}^s$ denote stars with centers $z_{j,0}$ and $q_{j,0}^s$, and leaves $z_{j,1}\,\dots,z_{j,4}$ and $q_{j,1}^s,\dots,q_{j,4}^s$, respectively.}\label{fig:fvc}
\end{figure}
For NP-hardness of \fvc, it remains to show that $G$ has a vertex cover $H 
\subset V$ satisfying $\left|H \cap N[v]\right| \leq k$ for every $v \in V$
if and only if the given \pname{planar 2P1N-3-Sat} instance has
a satisfying assignment. 

To see this, first assume that the given \pname{planar 2P1N-3-Sat}
instance has a satisfying assignment $\alpha\colon X\to\{0,1\}$ and
extend $\alpha$ to $\bar{X}$ by letting $\alpha(\bar{x}):=1-\alpha(x)$.
We construct the vertex cover $H=\{x\in X\cup\bar{X} \mid \alpha(x)=0\}\cup\bigcup_{i=1}^{n}\bigcup_{s=1}^{k-2}\{y_{i,0}^s\}\cup\bigcup_{j=1}^{m}\left(\{z_{j,0}\}\cup\bigcup_{s=1}^{k-|Cj|}\{q_{j,0}^{s}\}\right)$.
Indeed, $H$ is a vertex cover, since all edges are of the form $\{x,\bar{x}\}$
or incident to some $y_{i,0}^s$, $z_{j,0}$, or $q_{j,0}^{s}$.
Also, $\left|H \cap N[v]\right|=1$
for $v\in\bigcup_{i=1}^{n}\bigcup_{s=1}^{k-2}\bigcup_{r=1}^{k} \{y_{i,r}^{s}\}\cup\bigcup_{j=1}^{m}\bigcup_{r=1}^{k}\{z_{j,r}\}\cup\bigcup_{j=1}^{m}\bigcup_{s=1}^{k-|C_j|}\bigcup_{r=1}^{k}\{q_{j,r}^{s}\}$,
$\left|H \cap N[v]\right|=2$ for
$v\in\bigcup_{i=1}^{n}\bigcup_{s=1}^{k-2}\{y_{i,0}^{s}\}\cup\bigcup_{j=1}^{m}\bigcup_{s=1}^{k-|C_j|}\{q_{j,0}^s\}$,
and $\left|H\cap N[v]\right|=k$ for $v\in X\cup\bar{X}$. It
remains to consider $v=z_{j,0}$ for some $j\in\{1,\dots,m\}$. Observe
that $\sum_{x\in C_{j}}\alpha(x)\geq1$ since $\alpha$ is a satisfying
assignment. 
It holds that $\left|H \cap N[z_{j,0}]\right| = 1 + k-|C_j| + \sum_{x\in
C_{j}}\left|\{x\}\cap H\right| = 1 + k-|C_j| + \sum_{x\in C_{j}}(1-\alpha(x))\leq
1 + k-|C_j|+|C_j|-1=k$.
In either case, we conclude that $\left|H \cap N[v]\right| \leq k$.

Conversely, suppose that there exists a vertex cover $H \subseteq V$ with
$\left|H \cap N[v]\right|\leq k$ for all $v\in V$. Let $i\in\{1,\dots,n\}$ and
$s\in\{1,\dots,k-2\}$. Since $H$ is a vertex cover with $\left|H \cap
N[y_{i,0}^{s}]\right|\leq k$ and $\mathrm{deg}(y_{i,0}^{s})>k$,
we must have $y_{i,0}^{s}\in H$. Similarly, we must have $z_{j,0}\in H$
for all $j\in\{1,\dots,m\}$ and $q_{j,0}^{s}\in H$ for all
$j\in\{1,\dots,m\},s\in\{1,\dots,k-|C_j|\}$.
We claim that $\alpha(x)=\left|\{\bar{x}\}\cap H\right|$ defines
a satisfying assignment for the given \pname{planar 2P1N-3-Sat}
instance. To see this, first observe that $\left|\{x,\bar{x}\}\cap H \right|=1$
for all $x\in X$ since $H$ is a vertex cover and $\left|H \cap N[x] \right|\leq
k$. 
We can therefore extend $\alpha$ to $\bar{X}$ by setting
$\alpha(\bar{x})=1-\alpha(x)=\left|\{x\}\cap H\right|$.
Recall that for for all $j\in\{1,\dots,m\}$ we have $\left|H \cap N[z_{j,0}] 
\right|\leq k$ and \mbox{$\left\{z_{j,0},q_{j,0}^{1}, \dots, 
q_{j,0}^{k-|C_j|}\right\} \subseteq H$,} which implies $\left|H \cap 
\left(N[z_{j,0} \setminus \left\{z_{j,0},q_{j,0}^{1}, \dots, 
q_{j,0}^{k-|C_j|}\right\}\right)\right| \leq k-(1+k-|C_j|) = |C_j|-1$.
With this, for $j\in\{1,\dots,m\}$, we have $\sum_{x\in C_{j}}\alpha(x)=|C_j|-\sum_{x\in C_{j}}\alpha(\bar{x})=|C_j|-\sum_{x\in C_{j}}\left|\{x\}\cap H\right|\geq|C_j|-\left|\left(N[z_{j,0}]\setminus\left\{z_{j,0},q_{j,0}^{1}, \dots, q_{j,0}^{k-|C_j|}\right\}\right)\cap H\right|\geq |C_j| - (|C_j|-1) = 1$.
We conclude that $\alpha$ is a satisfying assignment.

We now turn to \ofvc and modify the above reduction as follows.
The first difference is that there 
is now only one $Y_i^s$ gadget for each variable $x_i$, so we call it $Y_i$. The 
second difference is that $Y_i$ is different from $Y_i^s$ from the previous 
reduction, because now $Y_i$ is responsible for picking only one vertex among 
$\{ x_i, \bar{x}_i \}$ to the solution. To be more precise $Y_i$ is now a star 
with a center $y_{i,0}$ connected to $y_{i,1}, \ldots y_{i,k-1}$, but now also 
each $y_{i,j},j\in \{1 \ldots, k-1 \}$ is a center of a star, connected to 
$y^j_{i,1}, \ldots, y^j_{i,k}$. In other words, $Y_i$ is a tree of depth two 
rooted at $y_{i,0}$, where the root has $k-1$ children and each child of the 
root has $k$ children. In the constructed graph, for each variable $x_i$ both 
literals $x_i$ and $\bar{x}_i$ are now connected to $y_{i,0}$ instead of 
$y_{i,0}^s$ vertices from the previous reduction. The remaining part of the 
graph is precisely the same as in the reduction for~\cref{thm:svc-hard}.  

Assume we have an \ofvc solution for the constructed graph and the given 
parameter $k \geq 4$. Observe, that for each $i$ the vertices $y_{i,0}$ and 
$y_{i,1} \ldots y_{i,k-1}$ must be taken to the solution since their degree is 
$k+1$. Therefore, since vertex $y_{i,0}$ has $k-1$ neighbors other than $x_i$ 
and $\bar{x}_i$, only one vertex among $x_i$ and $\bar{x}_i$ can be taken to the 
solution. The remaining part of the argument is the same as in the proof 
for \fvc: we construct a satisfying assignment by setting this 
literal to false, whose corresponding vertex was taken to the solution. 
Similarly as before, each clause has at most two negated literals. The opposite 
direction is analogous. 
\end{proof}

The previous reduction actually also shows that for any $\eps > 0$, it is NP-hard
to compute a $(4/3 - \eps)$-approximation for \fvc, as for $k=3$, a $(4/3 - 
\eps)$-approximate solution actually has sparseness $3$.
Still, we are able to compute a $(2 - \frac{1}{k})$-approximation for \fvc,
which is slightly better than our result for \svc and also better then the best
possible approximation ratio for \pname{Vertex Cover} (and thus \svc) under UGC.
In particular, our algorithm implies that for the smallest value $k=3$ for which 
\fvc is NP-hard, we can obtain a solution of sparseness $5$. We leave open 
whether a solution of sparseness $4$ can be computed in polynomial time for \fvc 
if~$k=3$, and whether better approximation algorithms are possible for \ofvc.

\fvcApprox*
\begin{proof}
As for the $2$-approximation algorithm for \svc (cf.\ \cref{thm:svc-approx}), 
we consider the relaxation of \eqref{shs-ILP}.
However, in order to improve the approximation ratio, observe that in any 
solution of cost $k$ to \fvc, every vertex of degree more than $k$ must 
be contained in the solution (otherwise some edge incident to such a vertex is 
not covered). Thus we may guess the optimum sparseness $k^\star$, define the set 
of high degree vertices $D=\{v\in V\mid\deg(v)>k^\star\}$ and add the constraint 
$x_v=1$ for every $v\in D$ to the above LP relaxation. We again let $W$ be the 
set of vertices with value at least $1/2$, which is a feasible vertex cover. If 
for any closed neighbourhood $N[v]\in\mc{B}$ we have $N[v]\subseteq W$ then all 
neighbours of $v$ are contained in~$W$, and thus we may remove $v$ from $W$ and 
still obtain a vertex cover. We repeat this iteratively for each vertex until we 
obtain a vertex cover $W$ for which no closed neighbourhood is entirely 
contained in $W$. In particular, for any $v\notin D$ we have $|W\cap N[v]|\leq 
k^\star$. For $v\in D$ on the other hand, since $x_v=1$ we get
\[
 |W\cap N[v]|\leq |D\cap N[v]|+2\sum_{u\in N[v]\setminus D} x_u\leq 
(2x_v-1)+2\sum_{u\in N[v]\setminus \{v\}} x_u\leq 2k-1.
\]
This means that the set $W$ yields a $(2-\frac{1}{k})$-approximation.
\end{proof}

\section{Hardness of Highway Dimension and Shortest Path Cover}

In this section, we study the parameterized complexity of \hd and \spc, and show
the following theorem.

\hdHard*

To prove \cref{thm:hd-hard}, we present a parameterized reduction from 
\pname{Clique} to \hd. 
This reduction also shows W[1]-hardness for \spc, as the constructed graph $G$ 
has radius at most $2r$, i.e., any solution for \hd 
is also a solution for \spc of the same cost and vice 
versa.

Let $H=(V,E)$ be a graph and let $k \in \mathbb{N}$.
Denote the number of vertices and edges of $H$ by $n$ and $m$, respectively.
For convenience we treat $H$ as a bidirected graph, i.e.\ we replace every edge $\{u,v\} \in E$ with directed edges $(u,v)$ and $(v,u)$.
Let $C$ be a constant whose value will be determined later on.
We construct a graph $G$ such that \hd has a solution of value $k' = 4C k(k-1) + 
\binom{k}{2} + k + 3$ on $G$ for $r = 2^m$ if and only if $H$ contains a clique 
if size $k$.
In the following, we call the individual elements of a solution for \hd also 
\emph{hubs}.

The graph $G$ contains $k(k-1)$ gadgets: For all $1 \leq i, j \leq k$ satisfying $i \neq j$ there is a gadget $G_{i,j}$. Choosing a certain set of hubs from $G_{i,j}$ means that $G_{i,j}$ represents a pair $(w_i,w_j)$ of adjacent vertices of $H$.
The idea of the reduction is to have a pair $(w_i,w_j)$ from every $G_{i,j}$ such that
\begin{enumerate}[label=(\roman*)]
		\item \label{cond:1} if $G_{i,j}$ represents $(x,y)$, then $G_{j,i}$ represents $(y,x)$, and
		\item \label{cond:2} if $G_{i,j}$ represents $(x,y)$ and $G_{i,j'}$ represents $(x',y')$, then $x = x'$.
\end{enumerate}
If these two conditions are fulfilled, it follows that there are $k$ distinct vertices $w_1, \dots, w_k$ which are pairwise adjacent, i.e.\ $\{w_1, \dots, w_k\}$ is a clique of size $k$.

Every gadget $G_{i,j}$ contains a path $u_{i,j}^1,\dots,u_{i,j}^{m}$, a path $v_{i,j}^1,\dots,v_{i,j}^{m}$, a path $a_{i,j}^1,\dots,a_{i,j}^{m}$, and a path $b_{i,j}^1,\dots,b_{i,j}^{m}$, each consisting of $m-1$ edges of length $1$.

We identify every vertex of these paths with a pair $(x,y)$ of adjacent vertices in $H$ as follows: 
Fix any ordering $\prec$ on $V$ and denote the resulting lexicographic ordering on $V \times V$ also by $\prec$.
Define $\tau \colon E \rightarrow \{1, \dots, m\}$ as
\[
	\tau(x,y) = \left| \{ (u,v) \in E \mid (u,v) \prec (x,y) \} \right|+1,
\]
i.e.\ $(x,y)$ is the $\tau(x,y)$-th edge according to $\prec$.
This allows us for instance to associate the vertex $u^{\tau(x,y)}_{i,j}$ of $G_{i,j}$ with the edge $(x,y)$ of $H$.

The four paths are connected as follows.
For $z \in \{a,v,b\}$ we connect $u_{i,j}^m$ with $z_{i,j}^1$ and $z_{i,j}^m$ with $u_{i,j}^1$, each through a path of length $r-m+3$.
To that end we introduce vertices $u_{i,j}^{0,z}, u_{i,j}^{m+1,z}, z_{i,j}^{0,u}$ and $z_{i,j}^{m+1,u}$, and add the edges $\{u_{i,j}^1,u_{i,j}^{0,z}\}, \{u_{i,j}^m,u_{i,j}^{m+1,z}\}, \{z_{i,j}^{1},z_{i,j}^{0,u}\}, \allowbreak \{z_{i,j}^m,z_{i,j}^{m+1,u}\}$ of length $1$ and the edges $\{u_{i,j}^{m+1,z}, z_{i,j}^{0,u}\}, \{z_{i,j}^{m+1,u}, u_{i,j}^{0,z}\}$ of length $r-m+1$.

Moreover, we add vertices $a_{i,j}^{m+1,v}, v_{i,j}^{0,a}, v_{i,j}^{m+1,b}, b_{i,j}^{0,v}$ and add edges $\{a_{i,j}^m, a_{i,j}^{m+1,v}\}, \{v_{i,j}^1, v_{i,j}^{0,a}\}, \allowbreak \{v_{i,j}^m, v_{i,j}^{m+1,b}\}, \allowbreak \{b_{i,j}^1, b_{i,j}^{0,v}\}$ of length $1$ and edges $\{a_{i,j}^{m+1},v_{0,a}\}, \{v_{i,j}^{m+1,b}, b_{i,j}^{0,v}\}$ of length $r-2m+2$.
This is illustrated in \cref{fig:gadget}.

\begin{figure}[p]
\centering
\includegraphics[scale=.75,page=1]{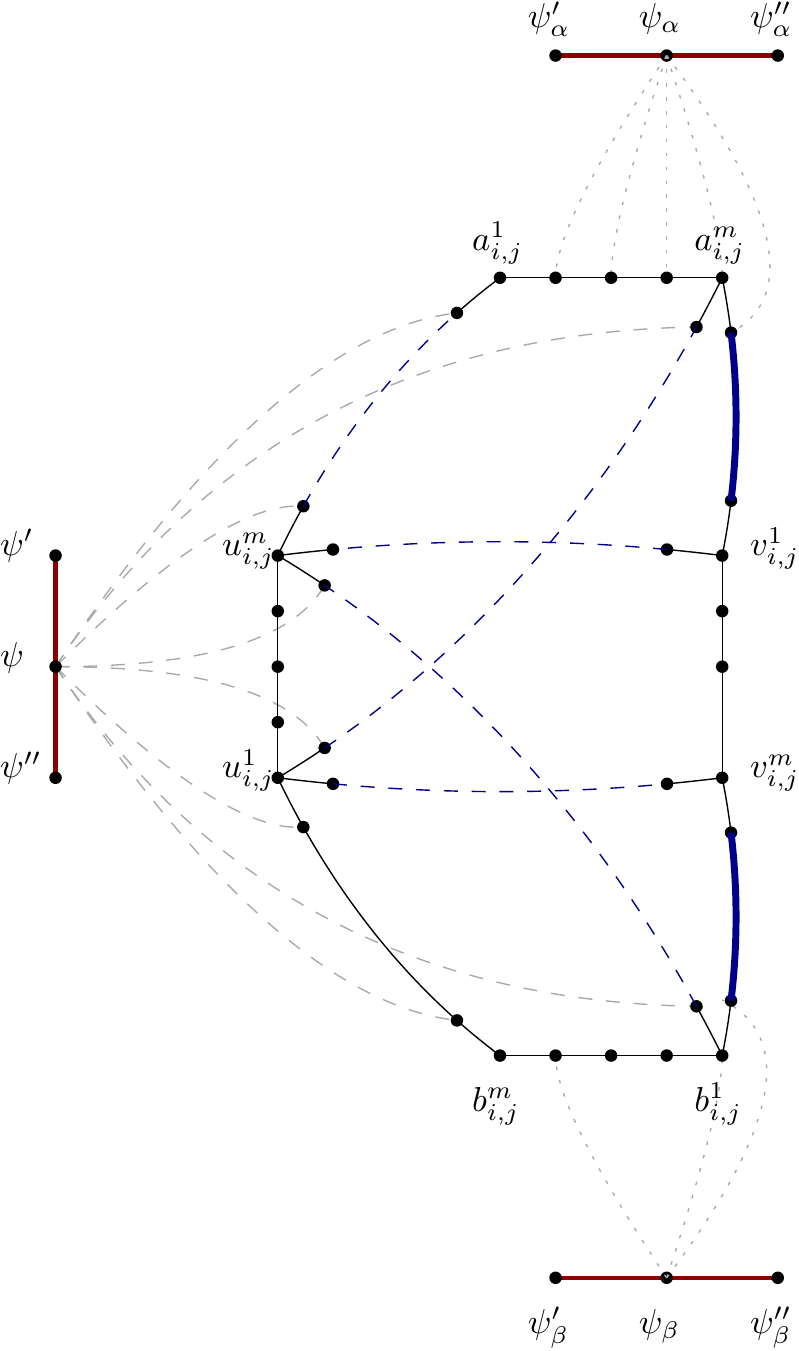}
\caption{A gadget $G_{i,j}$ and the global vertices. Solid black edges have 
length $1$, dashed blue edges have length $r-m+1$, thick blue edges have length 
$r-2m+2$, dotted gray edges have length $m-1$, dashed gray edges have length 
$r/2$ and red edges have length $r$.}\label{fig:gadget}
\vspace{1cm}
\includegraphics[scale=.5,page=2]{gadget-crop}
\caption{Two gadgets $G_{i,j}$ and $G_{j,i}$ and the connections between them. 
The marked vertices indicate that $G_{i,j}$ and $G_{j,i}$ represent the pairs 
$(x,y)$ and $(y,x)$, respectively. Moreover, the vertex $\alpha_{i,j}^{(x,y)}$ 
is marked.}\label{fig:gadgets_ij_and_ji}
\end{figure}

The idea is that the shortest path from $u_{j,j}^1$ to any of $a_{i,j}^{0,u}, 
v_{i,j}^{0,u}$, and $b_{i,j}^{0,u}$ has length $r+1$ and that we will have to 
choose some pair $(x,y)$ in order to hit these shortest paths through the hub 
$u_{i,j}^{\tau(x,y)}$.
Still, the shortest paths between $a_{i,j}^{0,u}$ and $b_{i,j}^{0,u}$ and between $a_{i,j}^{m+1,u}$ and $b_{i,j}^{m+1,u}$ both have length $2r - 2m + 4 > r$, but are not hit, if we choose, e.g., the hub $u_{i,j}^2$.
Hence, we introduce a shorter path between $a_{i,j}^{0,u}$ and $b_{i,j}^{0,u}$ and between $a_{i,j}^{m+1,u}$ and $b_{i,j}^{m+1,u}$, which will be hit by a global dummy hub: We add vertices $\psi,\psi',\psi''$ and the edges $\{\psi',\psi\}$ and $\{\psi'',\psi\}$, both of length $r$.
Moreover, we add edges between $\psi$ and $a_{i,j}^{0,u}, a_{i,j}^{m+1,u}, b_{i,j}^{0,u}$, and $b_{i,j}^{m+1,u}$, each of length $r/2$.
The shortest $a_{i,j}^{0,u}$-$b_{i,j}^{0,u}$- and $a_{i,j}^{m+1,u}$-$b_{i,j}^{m+1,u}$-paths now have length $r$ and  pass through $\psi$.
Furthermore, the shortest $\psi'$-$\psi''$-path has length $2r$ and we may assume w.l.o.g.\ that it is hit through the hub $\psi$.

We will show that if the final graph $G$ admits a solution of value $k'$, then there is a hitting set for $\mc{P}_r$ containing four vertices from every gadget $G_{i,j}$, which represent a pair $(x,y)$ of adjacent vertices of $H$.
Our construction needs to ensure that for these pairs $(x,y)$ conditions \ref{cond:1} and \ref{cond:2} are fulfilled.
First we create $C$ copies $G_{i,j}^1, \dots, G_{i,j}^C$ of the graph $G_{i,j}$. For simplicity, we confuse the graphs $G_{i,j}$ and $G_{i,j}^\lambda$ for $\lambda \in \{1, \dots, C\}$ when the context is clear.
Our final construction will yield that if $G_{i,j}^\lambda$ represents $(x,y)$ and $G_{i,j}^{\lambda'}$ represents $(x',y')$, then we have $(x,y) = (x',y')$.
Note that the vertices $\psi, \psi'$, and $\psi''$ are not part of any gadget $G_{i,j}$ and hence, we do not create copies of them.

For condition \ref{cond:1} we have to synchronise the gadgets $G_{i,j}$ and $G_{j,i}$. 
To that end, for all $1 \leq i < j \leq k$ and all $(x,y) \in E$ we add a vertex $\alpha_{i,j}^{(x,y)}$. Moreover, we add edges of weight $m$ from (all $C$ copies of) $a_{i,j}^{\tau(x,y)+1}$ and from (all $C$ copies of) $a_{j,i}^{\tau(y,x)+1}$ to $\alpha_{i,j}^{(x,y)}$.
This is illustrated in \cref{fig:gadgets_ij_and_ji}.
The idea is that all shortest paths between $G_{i,j}$ and $G_{j,i}$ contained in 
$\mc{P}_r$ can be hit with one additional hub $\alpha_{i,j}^{(x,y)}$ if both 
gadgets agree on the pairs $(x,y)$ and~$(y,x)$.

Still, the newly added edges of length $m$ add new shortest paths to 
$\mc{P}_r$. For instance, in any $G_{i,j}$, the shortest path between 
$u_{i,j}^{m+1,a}$ and $\alpha_{i,j}^{\tau^{-1}(1)}$ has length $r+2$. To ensure 
that it suffices to choose only $u_{i,j}^{\tau(x,y)}, v_{i_j}^{\tau(x,y)}$, and 
$\alpha_{i,j}^{(x,y)}$ as hubs, we remove these paths from $\mc{P}_r$ by 
creating a new shortest path between $u_{i,j}^{m+1,a}$ and 
$\alpha_{i,j}^{\tau^{-1}(1)}$, which passes through the dummy hub $\psi$. To 
that end, for all $1 \leq i , j \leq k$ satisfying $i \neq j$ we add an edge 
between $\psi$ and $u_{i,j}^{0,a}, u_{i,j}^{m+1,a}$ and all 
$\alpha_{i,j}^{(x,y)}$, each of length $r / 2$.

Similarly, we avoid ''undesired`` hubs covering shortest paths across different gadgets $G_{i,j}$ and $G_{j,i}$ by introducing new vertices $\psi_\alpha, \psi'_\alpha$ and $\psi''_\alpha$ and adding the edges $\{\psi'_\alpha, \psi_\alpha\}$ and $\{\psi''_\alpha, \psi_\alpha\}$ of length $r$ and an edge of length $m-1$ between $\psi_\alpha$ and all $a_{i,j}^{\tau(x,y)+1}$.

To fulfill condition \ref{cond:2} we have to synchronise the gadget $G_{i,j}$ with every other gadget $G_{i,j'}$. To that end, for all $1 \leq i \leq k$ and all $x \in V$ we add a vertex $\beta_i^{x}$. Let $y_0, \dots, y_d$ be the neighbors of $x$ such that $y_0 \prec \dots \prec y_d$.
For $1 \leq i,j \leq k, i \neq j$ we add an edge of weight $m+d$ between $\beta_i^{x}$ and every (copy of) $b_{i,j}^{\tau(x,y_0)-1}$. Here the idea is that if two gadgets $G_{i,j}$ and $G_{i,j'}$ represent pairs $(x,y)$ and $(x',y')$ such that $x = x'$, then choosing $\beta_i^x$ as a hub suffices to hit all relevant shortest paths between the two gadgets.

Again, we have to take care of newly created shortest paths. Therefore we add an edge of length $r/2$ between $\psi$ and $u_{i,j}^{0,b}, u_{i,j}^{m+1,b}$ and all $\beta_{i}^{x}$.
Moreover we handle shortest paths across different gadgets $G_{i,j}$ and $G_{i,j'}$ by introducing new vertices $\psi_\beta,\psi'_\beta$ and $\psi''_\beta$ and adding the edges $\{\psi'_\beta,\psi_\beta\}$ and $\{\psi''_\beta,\psi_\beta\}$ of length $r$ and an edge of length $m+d-1$ between $\psi_\beta$ and every $b_{i,j}^{\tau(x,y_d)-1}$ where $y_d$ is the maximum neighbor of $x$ according to $\prec$. This concludes the construction of the graph $G$, which is also illustrated in \cref{fig:whole_constuction}.

\begin{figure}[t]
\centering
\includegraphics[scale=.25]{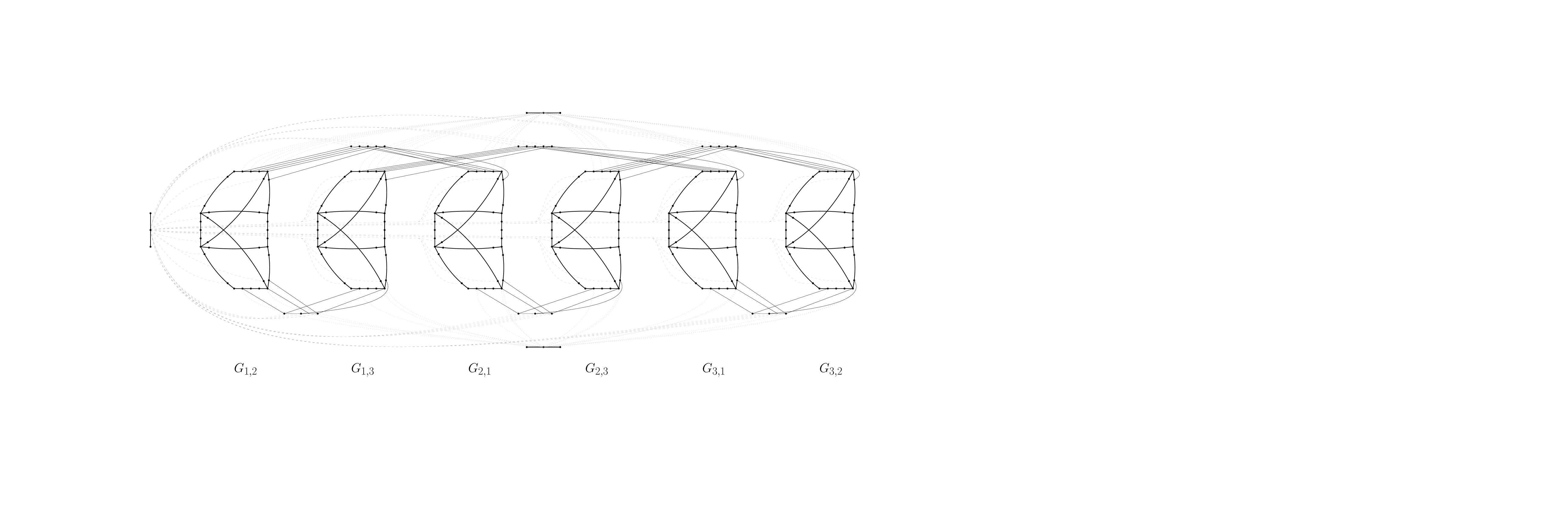}
\caption{A (simplified) illustration of the whole construction. Note that only one copy $G_{i,j}^{\lambda}$ of every $G_{i,j}$ is shown.}\label{fig:whole_constuction}
\end{figure}

We now show several properties of the graph $G$ which allow us to prove \cref{thm:hd-hard}.
The following \namecref{lem:gadget-represents-pair} states that choosing four hubs from some gadget $G_{i,j}$ means that the gadget represents a unique pair $(x,y)$.

\begin{lemma}\label{lem:gadget-represents-pair}
				Let $1 \leq i, j \leq k$ where $i \neq j$ and let $\mathcal H_{i,j}$ be a hitting set for all shortest paths from $\mc{P}_r$ that are contained in $G_{i,j}$. It holds that $\left| \mathcal H_{i,j} \right| \geq 4$ and moreover, if $\left| \mathcal H_{i,j} \right|=4$, then $\mathcal H_{i,j} = \left\{u_{i,j}^{\tau(x,y)}, a_{i,j}^{\tau(x,y)}, v_{i,j}^{\tau(x,y)}, b_{i,j}^{\tau(x,y)}\right\}$ for some pair $(x,y)$.
\end{lemma}

\begin{proof}
Let $1 \leq \iota \leq m+1$. For $z \in \{a,v,b\}$ we define the path $P^{uz}(\iota)$ as the shortest $s$-$t$-path path for
\[
s = \begin{cases} u_{i,j}^{\iota} & \text{if } \iota \leq m \\ u_{i,j}^{m+1,z} & \text{else}\end{cases}
\text{ and }
t = \begin{cases} z_{i,j}^{\iota-1} & \text{if } \iota > 1 \\ z_{i,j}^{0,u} & \text{else}\end{cases}.
\]

Similarly we define the path $P^{zu}(\iota)$ as the shortest $s$-$t$-path for
\[
s = \begin{cases} z_{i,j}^{\iota} & \text{if } \iota \leq m \\ z_{i,j}^{m+1,u} & \text{else}\end{cases}
\text{ and }
t = \begin{cases} u_{i,j}^{\iota-1} & \text{if } \iota > 1 \\ u_{i,j}^{0,z} & \text{else}\end{cases}.
\]

\begin{figure}[t]
	\centering
	\begin{subfigure}[b]{0.49\textwidth}
		\centering
		\includegraphics[scale=.7,page=3]{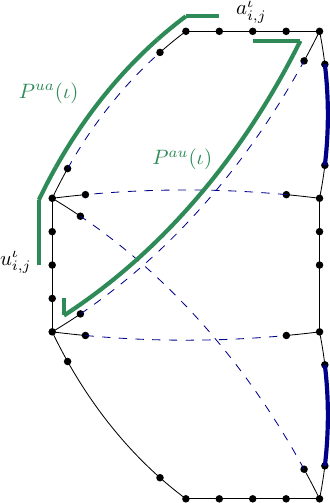}
		\caption{The shortest paths $P^{ua}(\iota)$ and $P^{au}(\iota)$.}\label{fig:path_Pua}
	\end{subfigure}
	\hfill
	\begin{subfigure}[b]{0.49\textwidth}
		\centering
		\includegraphics[scale=.7,page=4]{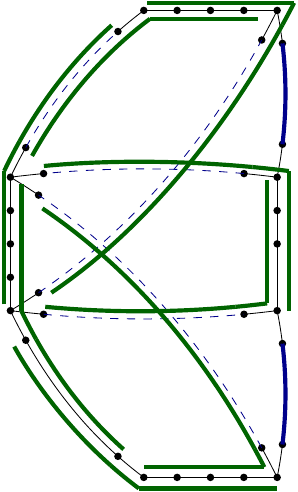}
		\caption{Eight shortest paths of length $r+1$.}\label{fig:eight_paths}
	\end{subfigure}
	\caption{Shortest paths in a gadget $G_{i,j}$}
\end{figure}

We observe that $P^{uz}(\iota)$ passes through the vertices $u_{i,j}^{m+1,z}$ and $z_{i,j}^{0,u}$ as illustrated in \cref{fig:path_Pua}, and that its length is $(m+1 - \iota) + (r-m+1) + (\iota-1) = r+1$. Similarly, $P^{zu}(\iota)$ passes through $z_{i,j}^{m+1,u}$ and $u_{i,j}^{0,z}$, and has also length $r+1$. This means that both $P^{uz}(\iota)$ and $P^{zu}(\iota)$ need to be hit by $\mathcal H_{i,j}$.
Consider the eight shortest paths $P^{ua}(1), P^{ua}(m+1), P^{au}(1), P^{uv}(m+1),\allowbreak P^{vu}(1),\allowbreak P^{bu}(1), P^{ub}(m+1)$, and $P^{bu}(m+1)$, which are shown in \cref{fig:eight_paths}. It holds that every vertex of $G_{i,j}$ covers at most two of these paths, which implies $\left| \mathcal H_{i,j} \right| \geq 4$.
To hit the shortest path $P^{uv}(m+1)$ we have to choose one of the vertices $u_{i,j}^{m+1,v}, v_{i,j}^{0,u}, v_{i,j}^{1}, \dots, v_{i,j}^m$. However, the vertices $u_{i,j}^{m+1,v}$ and $v_{i,j}^{0,u}$ do not hit any of the other seven shortest paths. Hence, if we have $\left| \mathcal H_{i,j} \right| = 4$, then one of the four hubs must be the vertex $v_{i,j}^{\tau(x,y)}$ for some $(x,y) \in E$.
Repeating the same argument for the paths $P^{au}(1), P^{ub}(m+1)$, and $P^{ua}(1)$, one can show that if $\left| \mathcal H_{i,j} \right| = 4$, then $\mathcal H_{i,j}$ consists of four vertices $v_{i,j}^{\tau(x,y)}, a_{i,j}^{\tau(x',y')}, v_{i,j}^{\tau(x'',y'')}, b_{i,j}^{\tau(x''',y''')}$.

We now show that for these four vertices it holds that $(x,y) = (x',y') = (x'',y'') = (x''',y''')$.
Suppose that for some $(x,y)$ we have $v_{i,j}^{\tau(x,y)} \in \mathcal H_{i,j}$.
Consider the two shortest paths $P^{uv}(\tau(x,y))$ and $P^{vu}(\tau(x,y))$. Our previous observations imply that both paths must be hit through some vertex $u_{i,j}^{\tau(x',y')}$. As $\mathcal H_{ij}$ contains precisely one vertex $u_{i,j}^{\tau(x',y')}$ and as both paths intersect only in $u_{i,j}^{\tau(x,y)}$, it follows that $u_{i,j}^{\tau(x,y)} \in \mathcal H_{i,j}$. Similarly, it follows that $\mathcal H_{i,j}$ needs to contain the vertices $a_{i,j}^{\tau(x,y)}$ and $b_{i,j}^{\tau(x,y)}$.
Hence, if $\left| \mathcal H_{i,j} \right| = 4$, it follows that there is a unique $(x,y)$ such that $\mathcal H_{i,j} = \left\{u_{i,j}^{\tau(x,y)}, a_{i,j}^{\tau(x,y)}, v_{i,j}^{\tau(x,y)}, b_{i,j}^{\tau(x,y)}\right\}$, and we say that $G_{i,j}$ represents $(x,y)$.
\end{proof}

Moreover, we show that if a gadget $G_{i,j}$ represents some pair $(x,y)$, then two certain shortest paths are not hit by the hubs of $G_{i,j}$.
To that end, for any $(x,y) \in E$ and all $1 \leq i < j \leq k$ and $\lambda \in \{1, \dots, C\}$, let $A_{i,j}^{(x,y),\lambda}$ be the shortest path between $\alpha_{i,j}^{(x,y)}$ and the $\lambda$-th copy of $v_{i,j}^{\tau(x,y)-1}$.
The length of $A_{i,j}^{(x,y),\lambda}$ is
\begin{multline*}
	\dist(\alpha_{i,j}^{(x,y)},a_{i,j}^{\tau(x,y)+1}) + 
\dist(a_{i,j}^{\tau(x,y)+1},a_{i,j}^{m+1,v}) + 
\dist(a_{i,j}^{m+1,v},v_{i,j}^{0,a}) + 
\dist(v_{i,j}^{0,a},v_{i,j}^{\tau(x,y)-1})  = \\
	m + m-\tau(x,y) + r-2m+2 + \tau(x,y)-1  = r+1.
\end{multline*}
Similarly, we define $B_{i,j}^{x,\lambda}$ as the shortest path between $\beta_i^{x}$ and the $\lambda$-th copy of $v_{i,j}^{\tau(x,y_d)+1}$, where $y_d$ is the maximum neighbor of $x$.
The path $B_{i,j}^{x,\lambda}$ consists of a $v_{i,j}^{\tau(x,y_d)+1}$-$v_{i,j}^{m+1,b}$-path of length $m-\tau(x,y_d)$, the edge $\{v_{i,j}^{m+1,b},b_{i,j}^{0,v}\}$ of length $r-2m+2$, a $b_{i,j}^{0,v}$-$b_{i,j}^{\tau(x,y_0)-1}$-path of length $\tau(x,y_d)-d-1$ and the edge $\{b_{i,j}^{\tau(x,y_0)-1}, \beta^x_i\}$ of length $m+d$.
The path $B_{i,j}^{x,\lambda}$ has length
\begin{multline*}
  \dist(\beta^x_i,b_{i,j}^{\tau(x,y_0)-1}) + 
\dist(b_{i,j}^{\tau(x,y_0)-1},b_{i,j}^{0,v}) + 
\dist(b_{i,j}^{0,v},v_{i,j}^{m+1,b}) + 
\dist(v_{i,j}^{m+1,b},v_{i,j}^{\tau(x,y_d)+1})  = \\
	m+d + \tau(x,y_d)-d-1 + r-2m+2 + m-\tau(x,y_d)  = r+1.
\end{multline*}
Moreover the following \namecref{lem:synchronizing-paths} holds, which is also illustrated in \cref{fig:lemma_paths}.

\begin{figure}[t]
\centering
\includegraphics[scale=.7,page=5]{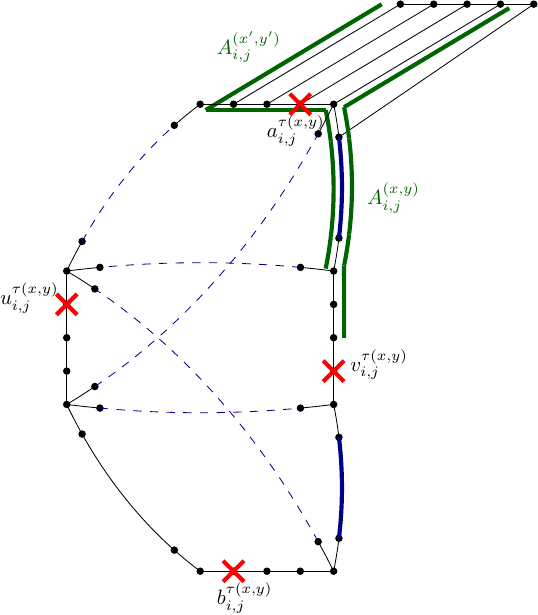}
\caption{An illustration of \cref{lem:synchronizing-paths}. The gadget $G_{i,j}$ represents $(x,y)$, which means that the shortest path $A_{i,j}^{(x,y)}$ is not hit by the hubs of  $G_{i,j}$, whereas any other shortest path $A_{i,j}^{(x',y')}$ is hit.}\label{fig:lemma_paths}
\end{figure}

\begin{lemma}\label{lem:synchronizing-paths}
If the gadget $G_{i,j}^\lambda$ represents the pair $(x,y)$, then the hubs of $G_{i,j}^\lambda$ hit
the shortest path $A_{i,j}^{(x',y'),\lambda}$ if and only if $(x,y) \neq (x',y')$, and
the shortest path $B_{i,j}^{x',\lambda}$ if and only if $x \neq x'$.
\end{lemma}
\begin{proof}
If the gadget $G_{i,j}^{\lambda}$ represents the pair $(x,y)$, then the hubs of $G_{i,j}^{\lambda}$ are $u_{i,j}^{\tau(x,y)}, v_{i,j}^{\tau(x,y)}, \allowbreak a_{i,j}^{\tau(x,y)}$, and $b_{i,j}^{\tau(x,y)}$.
The shortest path $A_{i,j}^{(x',y'),\lambda}$ contains the vertices $v_{i,j}^{\tau(x',y')-1}, \dots, v_{i,j}^{1}$ and the vertices $a_{i,j}^{m}, \dots, a_{i,j}^{\tau(x',y')+1}$. This means that $A_{i,j}^{(x',y'),\lambda}$ is hit if and only if $(x,y) \neq (x',y')$.
The shortest path $B_{i,j}^{x'}$ contains the vertices $v_{i,j}^{\tau(x',y_d)+1}, \dots, v_{i,j}^{m}$ and the vertices $b_{i,j}^{1}, \dots, a_{i,j}^{\tau(x',y_0)-1}$, which means that $B_{i,j}^{x'}$ is hit if and only if $x \neq x'$.
\end{proof}

Let us now prove \cref{thm:hd-hard}.
We show that on $G$, \hd has a solution of value $k' = 4C k(k-1) + \binom{k}{2} 
+ k + 3$ for $r = 2^m$ if and only if $H$ contains a clique if size $k$.

\begin{proof}[Proof of \cref{thm:hd-hard}]
Let $r = 2^m$. Suppose that on the constructed graph $G$, there is a solution of 
value $k'$ for \hd.
We observe that every vertex has distance less than $2r$ from the vertex $\psi$. This means that $B_{2r}(\psi)$ contains the entire graph, and therefore there is a hitting set $\mathcal H$ of size $\vert \mathcal H \vert \leq k'$ for $\mc{P}_r$.

We will prove that for any $1 \leq i,j \leq k$ there is some $(x,y)$ such that the hitting set $\mathcal H$ contains four hubs $u_{i,j}^{\tau(x,y)}, v_{i,j}^{\tau(x,y)}, a_{i,j}^{\tau(x,y)}, b_{i,j}^{\tau(x,y)}$ from every (copy of the) gadget $G_{i,j}$, and that $\mathcal H$ contains one hub $\alpha_{i,j}^{x,y}$ for every $1 \leq i < j \leq k$ and one hub $\beta_i^x$ for every $1 \leq i \leq k$, such that conditions \ref{cond:1} and \ref{cond:2} are satisfied.
This implies that $H$ contains a clique of size $k$.

Fix now $i,j$ such that $1 \leq i < j \leq k$. We prove that the hitting set $\mathcal H$ contains some hub $\alpha_{i,j}^{(x,y)}$.
Let $\lambda \in \{1, \dots, C\}$ and denote the vertices of $\mathcal H$ that are contained in $G_{i,j}^{\lambda}$ by $\mathcal H_{i,j}^{\lambda}$.
For the sake of contradiction suppose that there is no $(x,y)$ such that $\alpha_{i,j}^{(x,y)} \in \mathcal H$.
\cref{lem:gadget-represents-pair} states that if $\left| \mathcal H_{i,j}^{\lambda} \right| \leq 4$, then $\mathcal H_{i,j}^{\lambda} = \left\{u_{i,j}^{\tau(x,y)}, v_{i,j}^{\tau(x,y)}, a_{i,j}^{\tau(x,y)}, b_{i,j}^{\tau(x,y)}\right\}$ for some $(x,y)$ and therefore $\mathcal H$ does not hit the path $A_{i,j}^{(x,y),\lambda}$ according to \cref{lem:synchronizing-paths}.
Hence, we obtain that for all $\lambda \in \{1, \dots, C\}$ we have $\left| \mathcal H_{i,j}^{\lambda} \right| \geq 5$.
Moreover, \cref{lem:gadget-represents-pair} states that from any gadget $G_{i',j'}$, $(i',j') \neq (i,j)$ we have to choose at least four hubs.
This means however that $\left| \mathcal H \right| \geq 4Ck(k-1) + C$. If we choose $C = k^2$ we obtain that $4Ck(k-1) + C > k'$, so it cannot be that there is no hub $\alpha_{i,j}^{(x,y)} \in \mathcal H$.
Analogously we can show that $\mathcal H$ contains some hub $\beta_i^x$ for every $1 \leq i \leq k$, if $C = k^2$.
To that end, fix $1 \leq i \leq k$ and suppose that there is no $x$ such that $\beta_i^x \in \mathcal{H}$.
Again, it follows from \cref{lem:gadget-represents-pair,lem:synchronizing-paths} that for all $\lambda \in \{1, \dots, C\}$ we have $\left| \mathcal H_{i,j}^{\lambda} \right| \geq 5$, where $H_{i,j}^{\lambda}$ denotes the vertices of $\mathcal H$ that are contained in $G_{i,j}^{\lambda}$.
As we showed previously, for $C = k^2$ this means that $\left| \mathcal{H} \right| > k'$, and it follows that there must be some $x$ such that $\beta_i^x \in \mathcal{H}$.

This means that $\mathcal H$ contains at least one hub $\alpha_{i,j}^{(x,y)}$ for every $1 \leq i < j \leq k$, at least one hub $\beta_i^x$ for every $1 \leq i \leq k$ and at least four hubs from every $G_{i,j}^{\lambda}$. None of these hubs hits the shortest paths $\psi''-\psi-\psi''$, $\psi'_\alpha-\psi_\alpha-\psi''_\alpha$, or $\psi'_\beta-\psi_\beta-\psi''_\beta$. To hit these three paths, we need three additional hubs. As $\mathcal H$ has size at most $k' = 4Ck(k-1) + \binom{k}{2} + k + 3$, it follows that $\mathcal H$ contains precisely four hubs from every $G_{i,j}^{\lambda}$, so every gadget represents indeed a unique pair $(x,y)$. Moreover, for every $1 \leq i < j \leq k$ there is a unique hub $\alpha_{i,j}^{(x,y)}$ and for every $1 \leq i \leq k$ there is a unique hub $\beta_i^x$.

It remains to show that the pairs represented by the individual gadgets fulfill properties \ref{cond:1} and \ref{cond:2}.
Consider $i,j$ such that $1 \leq i < j \leq k$ and let $\lambda, \lambda' \in \{1, \dots, C\}$.
Let $(x,y)$ and $(x',y')$ be the pairs represented by $G_{i,j}^{\lambda}$ and $G_{j,i}^{\lambda'}$, respectively.
\cref{lem:synchronizing-paths} states that the hubs contained in $G_{i,j}^{\lambda}$ and $G_{j,i}^{\lambda'}$ do not hit the shortest paths $A_{i,j}^{(x,y),\lambda}$ and $A_{j,i}^{(x',y'),\lambda'}$. This means that the two paths must be hit through the hubs $\alpha_{i,j}^{(x,y)}$ and $\alpha_{i,j}^{(x',y')}$. Moreover, both hubs must coincide as $\mathcal H$ has size $k'$, i.e.\ we have $(x,y) = (x',y')$, which implies condition \ref{cond:1}.

For condition \ref{cond:2}, let $1 \leq i \leq k$, let $1 \leq j, j' \leq k$, and let $\lambda, \lambda' \in \{1, \dots, C\}$.
Denote the pairs represented by $G_{i,j}^{\lambda}$ and $G_{i,j'}^{\lambda'}$ by $(x,y)$ and $(x',y')$, respectively.
It follows from \cref{lem:synchronizing-paths} that the shortest paths $B_{i,j}^{x,\lambda}$ and $B_{i,j'}^{x',\lambda'}$ are not hit through the hubs contained in $G_{i,j}^{\lambda}$ and $G_{i,j'}^{\lambda'}$. This means that the paths must be covered through hubs $\beta_i^x$ and $\beta_i^{x'}$, and as $\left| \mathcal H \right| = k'$, this is only possible if $x = x'$, i.e.\ condition \ref{cond:2} is satisfied.

This implies that the graph $H$ indeed contains a clique of size $k$.

For the other direction suppose that the graph $H$ contains a clique $\{w_1, \dots, w_k\}$ of size $k$.
Consider the following set $\mathcal H$:
For $1 \leq i,j \leq k, i \neq j$ it contains all $C$ copies of the vertices $u_{i,j}^{\tau(w_i,w_j)}, v_{i,j}^{\tau(w_i,w_j)}, a_{i,j}^{\tau(w_i,w_j)}, b_{i,j}^{\tau(w_i,w_j)}$, for $1 \leq i < j \leq k$ it contains $\alpha_{i,j}^{(w_i,w_j)}$, for $1 \leq i \leq k$ it contains $\beta_{i,j}^{w_i}$, and moreover it contains the three vertices $\psi, \psi_\alpha, \psi_\beta$.
It holds that $\mathcal H$ has size~$k'$.

We can observe that all shortest paths between different gadgets $G_{i,j}$ and $G_{i',j'}$ are hit by $\psi_\alpha$ or $\psi_\beta$.
We now show that all shortest paths from $\mc{P}_r$ that intersect only one gadget $G_{i,j}$ are hit by $\mathcal H$.
Let $1 \leq i, j \leq k$ such that $i \neq j$. Suppose that $i < j$, the case $i > j$ can be shown similarly.
Consider some vertex $t$ contained in $G_{i,j}$ and denote the shortest path between $\alpha_{i,j}^{(x,y)}$ and $t$ by $P$.
Suppose that $P$ is not hit by $\psi$. We can observe that the shortest path between $\alpha_{i,j}^{(x,y)}$ and $u_{i,j}^{m+1,a}$ or $b_{i,j}^{0,u}$ contains $\psi$.
As $P$ is not hit by $\psi$, it follows that
\begin{enumerate}[label=(\alph*)]
		\item \label{case:P1} $t = a_{i,j}^{\iota}$ for some $\iota$ or $t \in \{a_{i,j}^{0,u}, a_{i,j}^{m+1,u}, a_{i,j}^{m+1,v}\}$,
		\item \label{case:P2} $t = v_{i,j}^{\tau(x',y')}$ for some $(x',y')$, or
		\item \label{case:P3} $t \in \{v_{i,j}^{0,a}, v_{i,j}^{0,u},v_{i,j}^{m+1,u}, v_{i,j}^{m+1,b}\}$.
\end{enumerate}
In case \ref{case:P1} it holds that $P$ has length \[\dist(\alpha_{i,j}^{(x,y)},a_{i,j}^{\tau(x,y)+1}) + \dist(a_{i,j}^{\tau(x,y)+1},t) \leq m+m+1 = 2m+1 < r,\] so it does not need to be hit by $\mathcal{H}$.
In case \ref{case:P2}, the length of $P$ is
\begin{align*}
				\dist(\alpha_{i,j}^{(x,y)},a_{i,j}^{\tau(x,y)+1}) + \dist(a_{i,j}^{\tau(x,y)+1},v_{i,j}^{0,a}) + \dist(v_{i,j}^{0,a},v_{i,j}^{\tau(x',y')}) & = \\
				m + m - \tau(x,y) + r-2m+2 + \tau(x',y') = r+2 + \tau(x',y') - \tau(x,y). &
\end{align*}
It follows that the length of $P$ exceeds $r$ if and only if $\tau(x',y') \geq \tau(x,y) - 1$, i.e.\ the path $P$ contains $A_{i,j}^{(x,y)}$ as a subpath.
\Cref{lem:synchronizing-paths} states that this subpath is hit by the hubs within $G_{i,j}$ if $(x,y) \neq (w_i,w_j)$, otherwise it is hit by $\alpha_{i,j}^{(x,y)}$.
Finally, in case \ref{case:P3} it holds that $P$ is shorter than $r$ or that $P$ contains $A_{i,j}^{(x,y)}$ as a subpath, which is hit by $\mathcal H$, as we just observed.

Analogously, consider some vertex $t$ contained in $G_{i,j}$, denote the shortest path between $\beta_i^x$ and $t$ by $P'$ and suppose that $P'$ is not hit by $\psi$. As the shortest path between $\beta_i^x$ and $u_{i,j}^{0,b}$ or $a_{i,j}^{m+1,u}$ contains $\psi$, it follows that
\begin{enumerate}[label=(\alph*)]
		\item \label{case:P'1} $t = b_{i,j}^{\iota}$ for some $\iota$ or $t \in \{b_{i,j}^{0,u}, b_{i,j}^{0,v}, b_{i,j}^{m+1,u}\}$,
		\item \label{case:P'2} $t = v_{i,j}^{\tau(x',y')}$ for some $(x',y')$, or
		\item \label{case:P'3} $t \in \{v_{i,j}^{0,a}, v_{i,j}^{0,u},v_{i,j}^{m+1,u}, v_{i,j}^{m+1,b}\}$.
\end{enumerate}
In case \ref{case:P'1} it holds that $P'$ is shorter than $r$.
In case \ref{case:P'2}, the length of $P'$ is
\begin{align*}
	\dist(\beta_i^x,b_{i,j}^{\tau(x,y_0)-1}) + \dist(b_{i,j}^{\tau(x,y_0)-1},v_{i,j}^{m+1,b}) + \dist(v_{i,j}^{m+1,b},v_{i,j}^{\tau(x',y')}) & = \\
	m + d + \tau(x,y_0)-1 + r-2m+2 + m+1-\tau(x',y') & = \\
	r + 2 + \tau(x,y_0) + d - \tau(x',y') = r + 2 + \tau(x,y_d) - \tau(x',y'). &
\end{align*}

It holds that the length of $P'$ exceeds $r$ if and only if $\tau(x',y') \leq \tau(x,y_d) + 1$, which is the case if and only if $P'$ contains $B_{i,j}^x$ as a subpath, which is hit by the hubs within $G_{i,j}$ or by $\beta_i^{(w_i,w_j)}$.
In case \ref{case:P'3} it holds that the length of $P'$ is at most $r$ or that $P'$ contains $B_{i,j}^x$ as a subpath, which is hit by $\mathcal{H}$.

Consider now the shortest path between $u_{i,j}^{\tau(x,y)}$ and some vertex $t$ of $G_{i,j}$ and denote it by $P''$.
Define the shortest paths $P^{uz}(\iota)$ and $P^{zu}(\iota)$ as in the proof of \cref{lem:gadget-represents-pair}.
If the length of $P'$ exceeds $r$ then for some $z \in \{a,v,b\}$, the path $P''$ contains the path $P^{uz}(\tau(x,y))$ to $z_{i,j}^{\tau(x,y)-1}$ or the path $P^{zu}(\tau(x,y)+1)$ to $z_{i,j}^{\tau(x,y)+1}$ as a subpath.
Suppose that $P''$ contains $P^{uz}(\tau(x,y))$, the other case is analogous.
To show that $P^{uz}(\tau(x,y))$ (and hence also $P''$) is hit by $\mathcal H$, we distinguish two cases:
If $(w_i,w_j) \prec (x,y)$, then $P^{uz}(\tau(x,y))$ is hit through $z_{i,j}^{\tau(w_i,w_j)-1}$, otherwise it is hit through $u_{i,j}^{\tau(w_i,w_j)-1}$.
Similarly it can be shown that any shortest path between two vertices of 
$G_{i,j}$ whose length exceeds $r$ is hit by $\mathcal H$, which means that 
$\mathcal{H}$ is a solution for \hd on $G$ of value $k'$.
\end{proof}

\section{Approximating Shortest Path Covers}
In this section, we show how to approximate \spc.

\hdApprox*

We present an algorithm based on the following ideas.
It is well-known that the \pname{Set Cover} problem is 
equivalent to \pname{Hitting Set} by swapping the roles of the elements of the 
universe and the sets in the given set family. 
\citet{DBLP:conf/cocoon/KuhnRWWZ05} study the \pname{Minimum Membership Set 
Cover (MMSC)} problem, where the aim is to minimize the maximum 
\emph{membership} of any element of the given universe of the \pname{Set Cover} 
instance. Here the \emph{membership} of an element is the number of sets of the 
solution it is contained in. The \pname{MMSC} problem finds applications in 
interference minimization in cellular networks, and 
\citet{DBLP:conf/cocoon/KuhnRWWZ05} prove that it admits a polynomial-time 
$O(\log|U|)$-approximation, where $U$ is the given universe, and they show that 
this is best possible, unless P=NP. Translated to \shs, this means that for an 
instance where $\mc{F}=\mc{B}$, an $O(\log|\mc{F}|)$-approximation can be 
computed in polynomial time, and this is also best possible, unless P=NP. We 
show that \spc can be reduced to this version of \shs.

We first give a simple observation about the \shs\ problem which will be useful 
later in our proof. Let $(V, \mc{F}, \mc{B})$ be a set system and let $B, B' 
\in \mc{B}$ be two sets such that~$B \subsetneq B'$. If $B'$ contains at 
most~$k$ elements of the hitting set, then $B$ also contains at most~$k$ such 
elements. Hence we obtain the following.

\begin{observation}\label{obs:ignoreB}
Let $\mc{B}$ be a family containing two sets $B,B'$ such that $B 
\subsetneq B'$. If there exists a solution to \shs for $(V, \mc{F}, \mc{B} 
\setminus \{B\})$ of sparseness $k$, then there exists a solution to \shs for 
$(V, \mc{F}, \mc{B})$ of sparseness $k$.
\end{observation}

We reduce the \spc problem to the \pname{Minimum Membership Set Cover} 
(\pname{MMSC}) problem. Formally, an instance of \pname{MMSC} consists of a 
universe $U$ and a family $\mc{S}$ of subsets of $U$, and the goal is to choose 
a set $\mc{S}' \subseteq \mc{S}$ such that every element in $U$ belongs to at 
least one set in $\mc{S}'$ and that the maximum membership of any element $u$ 
with respect to $\mc{S}'$ is minimal, where the membership of $u$ is defined as 
the number of sets in $\mc{S}'$ containing $u$. 

Recall that, given a weighted graph $G=(V,E)$ and a radius $r >0$, the \spc 
problem for $G$ is equivalent to the \shs problem on universe $V$ with $\mc{F} = 
\mc{P}_r$ and $\mc{B}=\{B_{2r}(v)\mid v\in V\}$. Based on \cref{obs:ignoreB}, we 
first show that if there exists a ball $B \in \mc{B}$ which does not contain any 
shortest path in $\mc{P}_r$ completely,  we can safely remove it without 
affecting the solution.

\begin{lemma}\label{lem:ignoreB}
Let $B \in \mc{B}$ which does not contain any shortest path in $\mc{P}_r$ 
completely, i.e., $S \nsubseteq B$ for every $S \in \mc{P}_r$. If there 
exists a solution to \spc for $(V,\mc{P}_r,\mc{B}\setminus \{B\})$ of 
sparseness~$k$, then there exists a solution to \spc for $(V, \mc{P}_r, 
\mc{B})$ of sparseness $k$.
\end{lemma}

\begin{proof}
First, if $S \cap B=\emptyset$ for every $S \in \mc{P}_r$, then there exists a 
solution for $(V,\mc{P}_r,\mc{B})$ not containing any vertices of $B$, and the 
claim follows. Now assume there is some path $S_B \in \mc{P}_r$ intersecting $B$ 
in some vertex $w$. We show that there exists a ball $B' \in \mc{B}$ such 
that~$B \subsetneq B'$, and thus the lemma follows from \cref{obs:ignoreB}.

Let $v$ be the center of the ball $B=B_{2r}(v)$. As $S \nsubseteq B$ for 
every $S \in \mc{P}_r$, $\dist(u,v) \leq r$ for every $u \in B$, as otherwise 
the shortest $u$-$v$-path would be contained in the ball $B$ of radius~$2r$ 
with a length in~$(r,2r]$, which would then be in $\mc{P}_r$.
Hence for any two vertices $u,u'\in B$, $\dist(u,u') \leq \dist(u,v) + 
\dist(v,u') \leq 2r$, so we have $B \subseteq B_{2r}(w)$.
Moreover, it holds that $S_B \subseteq B_{2r}(w)$ as $S_B$ is the vertex set of 
a path containing $w$ of length in $(r,2r]$, and it holds that $S_B \not 
\subseteq B$, which implies $B \subsetneq B_{2r}(w)$. By definition of $\mc{B}$ 
we have $B_{2r}(w) \in \mc{B}$, and thus by \cref{obs:ignoreB} the lemma 
follows.
\end{proof}

\cref{lem:ignoreB} means that we may assume w.l.o.g.\ that for any $B \in 
\mc{B}$ there is some $S_B \in \mc{P}_r$ such that $S_B \subseteq B$. We now 
give the following observations about the relationship among $\mc{P}_r$, 
$\mc{B}$, and a hitting set $H$ of $\mc{P}_r$. 

\begin{observation}\label{obs:ballOfSet}
Let $S \in \mc{P}_r$. As $S$ is the set of vertices of a shortest path $\pi$ of 
length $\ell(\pi) \in (r,2r]$, there exists a ball $B_S \in \mc{B}$ of radius 
$2r$, which completely contains $S$. This, in turn, implies that ${H} \cap B_S 
\neq \emptyset$  and $|H \cap S| \leq |H \cap B_S|$.
\end{observation}

\begin{observation}\label{obs:setOfBall}
Let $B \in \mc{B}$. If $B$ contains some shortest path set $S_B \in \mc{P}_r$,
then we have $H \cap B \neq \emptyset$ and $|H \cap S_B| \leq |H \cap B|$.
\end{observation}
By \cref{obs:ballOfSet,obs:setOfBall}, we get the following.
\begin{lemma}
There exists a solution to \shs\ for $(V,\mc{P}_r,\mc{B})$ of sparseness~$k$ 
 if and only if there exists a solution to \shs\ for $(V, \mc{P}_r \cup 
\mc{B} ,\mc{B} \cup \mc{P}_r)$ of sparseness~$k$.
\end{lemma}

\begin{proof}
Observe that any solution to \shs\ for $(V, \mc{P}_r \cup \mc{B} ,\mc{B} 
\cup \mc{P}_r)$ is also a solution to \shs\ for $(V,\mc{P}_r,\mc{B})$, as 
$\mc{P}_r \subseteq \mc{P}_r \cup \mc{B}$ and $\mc{B} \subseteq \mc{B} \cup 
\mc{P}_r$. We now prove that any solution $H$ to \shs\ for 
$(V,\mc{P}_r,\mc{B})$ of sparseness $k$ is also a solution to  \shs\ for 
$(V, \mc{P}_r \cup \mc{B} ,\mc{B} \cup \mc{P}_r)$ of sparseness~$k$. For this, 
we need to show that for every $S \in \mc{P}_r$, $|H \cap S| \leq 
k$ and for every $B \in \mc{B}$, $H \cap B \neq \emptyset$. The former statement
follows from \cref{obs:ballOfSet}, while the latter follows from 
\cref{obs:setOfBall} where we assume that $B$ contains some $S_B\in\mc{P}_r$ 
due to \cref{lem:ignoreB}.
\end{proof}

We now define an instance of the \pname{Minimum Membership Set Cover} with $U = 
\mc{P}_r \cup \mc{B}$ and $\mc{S} = \{S_u \mid u \in V\}$, where $S_u = \{S \in 
U \mid u \in S\}$, and prove the following.

\begin{lemma}
There exists a solution to \shs for $(V, \mc{P}_r \cup \mc{B} ,\mc{B} \cup 
\mc{P}_r)$ of sparseness~$k$ if and only if there exists a solution to 
\pname{MMSC} for $(U, \mc{S})$ of value~$k$.
\end{lemma}

\begin{proof}
We will prove that if there exists a solution to \shs for $(V, \mc{P}_r \cup 
\mc{B} ,\mc{B} \cup \mc{P}_r)$ of sparseness $k$, then there exists a solution 
to \pname{MMSC} for $(U, \mc{S})$ of value $k$. The proof for the other 
direction is symmetric. Let $H$ be a solution to \shs for $(V, \mc{P}_r \cup 
\mc{B} ,\mc{B} \cup \mc{P}_r)$ of sparseness $k$. We claim that the set $W 
= \{S_u \in \mc{S}~|~u \in H\}$ is a solution to \pname{MMSC} for~$(U, \mc{S})$ 
of value $k$. Let $E \in U$. Then, $H \cap E \neq \emptyset$. Let $u \in H \cap 
E$. By the definition of $S_u$, this implies that $E \in S_u$. Moreover, for any 
$B \in \mc{P}_r \cup \mc{B}$, we have that $|H \cap B| \leq k$. This implies 
that $B$ belongs to at most $k$ sets in $W$. Hence,  $W$ is a solution to 
\pname{MMSC} for~$(U, \mc{S})$ of value~$k$. 
\end{proof}

Since there exists a $O(\log |U|)$-approximation algorithm for \pname{MMSC} by 
Kuhn et al.~\cite{DBLP:conf/cocoon/KuhnRWWZ05} and $|U| = |\mc{P}_r \cup \mc{B}| 
= O(n^2)$, by the above lemma we get an $O(\log n)$-approximation algorithm 
for~\spc. This concludes the proof of \cref{thm:hd-approx}.

\section{Dense Matching}
Finally, we consider the \pname{Dense Matching} problem and prove the following theorem.

\matchingHard*

\begin{proof}
Consider the following reduction from \textsc{$3$-Sat}.
Let an instance of this problem be given by a set of variables $X=\{x_{i}\}_{i=1,\dots,n}$
and a set of clauses $\mathcal{C}=\{C_{j}\}_{j=1,\dots,m}$ with $C_{j}\subset 
X\cup\bar{X},|C_{j}| \leq 3$. Let $\bar{\bar{x}} = x$.
We construct the graph $G = (V,E)$ given by
\[
	V = \bigcup_{i=1}^n \left( \{ x_i, \bar{x}_i, x^0 \} \cup \{ x_i^{\ell}, \bar{x}_i^{\ell} \mid 1 \leq \ell \leq 7\} \right) \cup \bigcup_{j=1}^m \left( \{ z_j \} \cup \bigcup_{x \in C_j} \{ x^{j,\ell} \mid 1 \leq \ell \leq 4 \} \right)
\]
and
\begin{align*}
	E = & \bigcup_{i=1}^n \left( \{ \{x_i, x_i^0\}, \{x_i,x_i^1\}, \{x_i^1,x_i^0\},\{\bar{x}_i, x_i^0\}, \{\bar{x}_i,\bar{x}_i^1\}, \{\bar{x}_i^1,x_i^0\} \} \right) \cup \\
				 & \bigcup_{i=1}^n \left( \bigcup_{\ell=1}^6 \{ \{x_i^{\ell}, x_i^{\ell+1}\}, \{\bar{x}_i^{\ell}, \bar{x}_i^{\ell+1}\} \} \cup \{\{x_i^7, x_i^4\}, \{\bar{x}_i^7, \bar{x}_i^4\} \} \right) \cup \\
				 & \bigcup_{j=1}^m \bigcup_{x \in C_j} \{ \{z_j,x^{j,1}\}, \{x^{j,1},x^{j,2}\}, \{x^{j,2},x^{j,3}\}, \{x^{j,3},x^{j,4}\}, \{x^{j,4},x\}, \{x^{j,4},x^0\} \}
\end{align*}
The construction is illustrated in \cref{fig:matching}.

\begin{figure}[t]
\centering
\includegraphics[width=\textwidth]{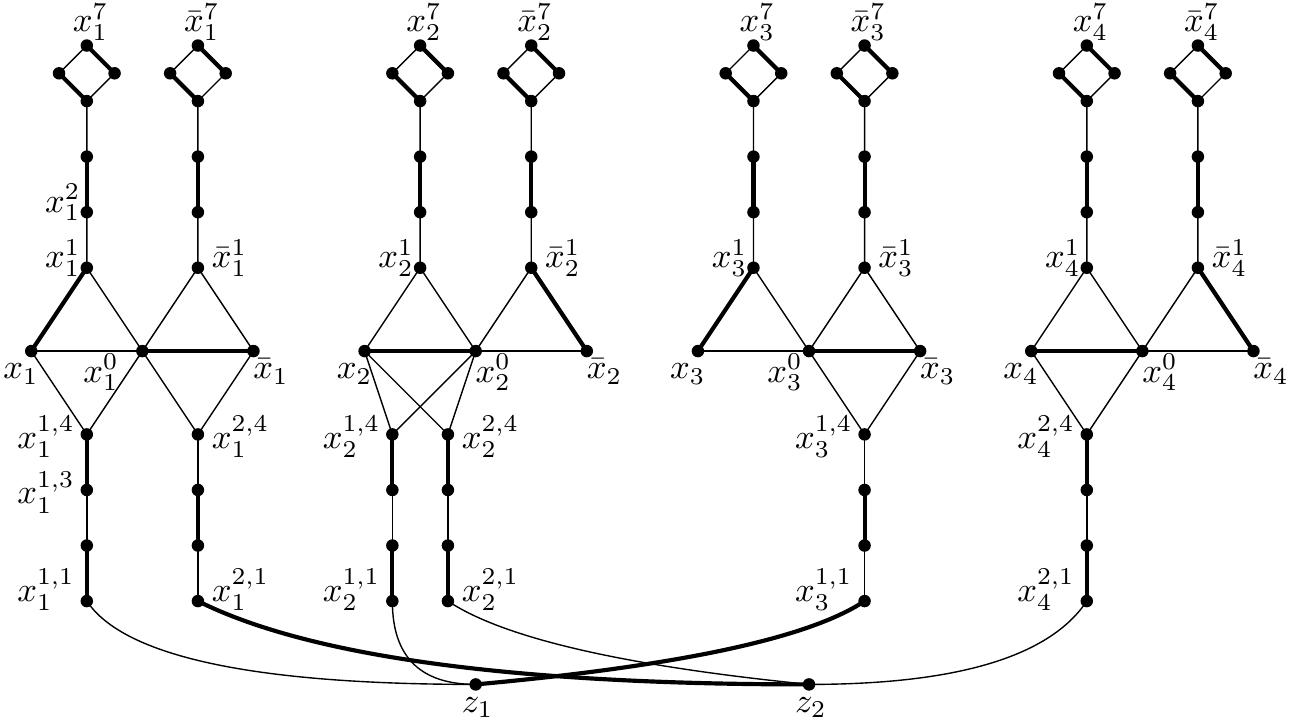}
\caption{The graph $G$ for the formula $(x_1 \vee x_2 \vee \bar{x}_3) \wedge (\bar{x}_1 \vee x_2 \vee x_4)$. The bold edges yield a matching that corresponds 
to the assignment $x_1 \mapsto 0, x_2 \mapsto 1, x_3 \mapsto 0, x_4 \mapsto 1$.}\label{fig:matching}
\end{figure}

We now show that the given \textsc{$3$-Sat} formula is satisfiable if and only 
if there is a matching~$M$ such that $|M \cap E(B_2(v))| \geq 2$ for every ball 
$B_2(v)$ of radius $2$, where we assume that edges have unit length.
This means that if the given formula is not satisfiable, then there is a ball $B_2(v)$ such that $|M \cap E(B_2(v))| \leq 1$, which implies that it is NP-hard to obtain an approximation factor less than two.

Suppose that the given formula has a satisfying assignment $\alpha \colon X \rightarrow \{0,1\}$ and extend $\alpha$ to $\bar X$ by choosing $\alpha(\bar{x}) = 1 - \alpha(x)$.
For $j=1 \dots m$ let $y_j \in C_j$ be some literal satisfying $C_j$, i.e. $\alpha(y_j) = 1$.
We construct the matching
\begin{align*}
    M = & \bigcup_{i=1}^n \{ \{x_i^2,x_i^3\}, \{x_i^4,x_i^5\}, \{x_i^6,x_i^7\}, \{\bar{x}_i^2,\bar{x}_i^3\}, \{\bar{x}_i^4,\bar{x}_i^5\}, \{\bar{x}_i^6,\bar{x}_i^7\}, \} \cup \bigcup_{x \colon \alpha(x)=1} \{ \{x,x^0\} \} \\
		\cup & \bigcup_{x \colon \alpha(x) = 0} \{ \{x,x^1\} \} \cup \bigcup_{j=1}^m \left(\{ \{z_j,y_j^{j,1}\}, \{y_j^{j,2},y_j^{j,3} \}\} \cup \bigcup_{x \in C_j \setminus \{y_j\}} \{ \{x^{j,1},x^{j,2}\}, \{x^{j,3},x^{j,4}\} \}\right)
\end{align*}
It is easy to verify that $M$ is indeed a matching and that $|M \cap E(B_2(v))| 
\geq 2$ for every ball~$B_2(v)$.

Suppose now that there is some matching $M$ such that $|M \cap E(B_2(v))| \geq 2$ for every ball $B_2(v)$.
Consider the assignment $\alpha \colon X \rightarrow \{0,1\}, \alpha(x) = 1$ if and only if there is some $j$ such that $\{z_j, x^{j,1}\} \in M$.
To show that $\alpha$ is a satisfying assignment, consider some clause $C_j$ and let $C_j = \{x_{i_1}, x_{i_2}, x_{i_3} \}$.
Consider the ball $B_2(x_{i_1}^{j,1}) = \{ z_j, x_{i_1}^{j,1}, x_{i_1}^{j,2}, x_{i_1}^{j,3}, x_{i_2}^{j,1}, x_{i_3}^{j,1} \}$.
We show that $M$ contains one of the edges $\{ z_j, x_{i_1}^{j,1}\}, \{ z_j, x_{i_2}^{j,1}\},$ and $\{ z_j, x_{i_3}^{j,1}\}$.
To prove this, suppose that $M$ contains none of these three edges.
This means that $M$ has to contain the two remaining edges $\{x_{i_1}^{j,1}, x_{i_1}^{j,2}\}$ and $\{x_{i_1}^{j,2}, x_{i_1}^{j,3}\}$ contained in $E(B_2(x_{i_1}^{j,1}))$, which is not possible as both edges are incident to $x_{i_1}^{j,2}$.

Let now $\{ z_j, x_{\iota}^{j,1}\}$ be the edge contained in $M$.
If $x_{\iota} \in X$, i.e. $x_{\iota}$ is a positive literal, it immediately follows that $\alpha$ satisfies the clause $C_j$.
Suppose now that $x_{\iota} \in \bar{X}$. We show that in this case we have $\alpha(x_{\iota}) = 0$, i.e.\ there is no $j'$ such that $\{z_{j'}, \bar{x}_{\iota}^{j',1}\} \in M$.
For the sake of contradiction, suppose that $\{z_{j'}, \bar{x}_{\iota}^{j',1}\} \in M$ for some $j' \in \{1,\dots,m\}$.
It follows from $|M \cap E(B_2(x_{\iota}^{j,1}))| \geq 2$ that $\{ x_{\iota}^{j,2}, x_{\iota}^{j,3} \} \in M$.
Consider the ball $B_2(x_{\iota}^{j,3}) = \{ x_{\iota}^{j,1}, x_{\iota}^{j,2}, x_{\iota}^{j,3}, x_{\iota}^{j,4}, x_{\iota}, \bar{x}_{\iota}^0 \}$.
As $M$ contains two edges from this ball and one of these edges is $\{ x_{\iota}^{j,2}, x_{\iota}^{j,3} \}$, the other edge needs to be contained in the triangle $\{x_{\iota}^{j,4}, x_{\iota}, \bar{x}_{\iota}^0 \}$.
Consider now the ball $B_2(x_{\iota}^7) = \{x_{\iota}^4, x_{\iota}^5, x_{\iota}^6, x_{\iota}^7\}$.
It holds that $\{x_{\iota}^4, x_{\iota}^5\} \in M$ or $\{x_{\iota}^4, x_{\iota}^6\} \in M$, which implies $\{x_{\iota}^3, x_{\iota}^4\} \not \in M$.
If we now consider the ball $B_2(x_{\iota}^2) = \{x, \bar{x}_{\iota}^0, x_{\iota}^1, x_{\iota}^2, x_{\iota}^3, x_{\iota}^4\}$, it follows that $M$ needs to contain one edge from the triangle $\{x, \bar{x}_{\iota}^0, x_{\iota}^1\}$.
As $M$ also contains one edge from the triangle $\{x_{\iota}^{j,4}, x_{\iota}, \bar{x}_{\iota}^0 \}$, we obtain that $M$ contains $\{\bar{x}_{\iota}^0, x_{\iota}^4\}, \{\bar{x}_{\iota}^0, x\}$, or $\{\bar{x}_{\iota}^0, x_{\iota}^1\}$.

However, as we also have $\{z_{j'}, \bar{x}_{\iota}^{j',1}\} \in M$, it follows analogously that $M$ contains $\{\bar{x}_{\iota}^0, \bar{x}_{\iota}^4\}, \{\bar{x}_{\iota}^0, \bar{x}\}$, or $\{\bar{x}_{\iota}^0, \bar{x}_{\iota}^1\}$, which is not possible.
This means that $\alpha$ is indeed a satisfying assignment, which concludes the proof.
\end{proof}

\printbibliography

\end{document}